\documentclass[11pt,letterpaper]{article}

\usepackage[margin=1.15in]{geometry}
\usepackage[T1]{fontenc}
\usepackage{lmodern}

\usepackage{microtype}
\usepackage{graphicx}
\usepackage{subfigure}
\usepackage{booktabs} 

\usepackage{amsfonts}
\usepackage{MnSymbol}
\usepackage{nicefrac}

\usepackage{tikz}
\usepackage{pgfplots}
\pgfplotsset{compat=1.18}
\usetikzlibrary{fpu} 

\let\PolishL\L
\usepackage{physics}
\usepackage{bm}
\usepackage{bbm}
\usepackage{xspace}
\usepackage{complexity}
\usepackage{mleftright}
\usepackage{xcolor}

\newcommand{\textred}[1]{{\color{red}#1}}

\mleftright

\def\vn{{\bm{n}}}

\def\vp{{\bm{p}}}
\def\vq{{\bm{q}}}

\def\vu{{\bm{u}}}
\def\vv{{\bm{v}}}
\def\vw{{\bm{w}}}
\def\vx{{\bm{x}}}
\def\vy{{\bm{y}}}

\renewcommand\vec\bm

\def\cK{{\mathcal{K}}}
\def\cL{{\mathcal{L}}}

\def\cX{{\mathcal{X}}}

\newcommand{\ftrl}{\texttt{FTRL}}
\newcommand{\md}{\texttt{MD}}
\newcommand{\mwu}{\texttt{MWU}}

\newcommand{\KKTgap}{\mathsf{KKTGap}}

\newcommand{\delimit}[3]{\newcommand{#1}[1]{\left#2##1\right#3}}
\delimit \ceil \lceil \rceil
\delimit \floor \lfloor \rfloor

\DeclareMathOperator*{\argmin}{argmin}
\DeclareMathOperator*{\argmax}{argmax}

\renewcommand{\R}{\mathbb R}
\newcommand{\N}{\mathbb N}

\newcommand{\defeq}{\coloneqq}

\let\L\PolishL

\usepackage{amsthm}
\usepackage{amsmath}
\usepackage{mathtools}

\usepackage{thm-restate}

\newcommand{\CLS}{\class{CLS}}

\usepackage{natbib}

\usepackage{hyperref}
\definecolor{linkblue}{HTML}{4F8FBA}
\hypersetup{colorlinks, citecolor=linkblue, linkcolor=linkblue, urlcolor=linkblue}
\usepackage[capitalize,noabbrev,nameinlink]{cleveref}
\theoremstyle{plain}
\newtheorem{theorem}{Theorem}[section]

\newtheorem{lemma}[theorem]{Lemma}
\newtheorem{claim}[theorem]{Claim}
\newtheorem{corollary}[theorem]{Corollary}
\theoremstyle{definition}
\newtheorem{definition}[theorem]{Definition}

\theoremstyle{remark}

\crefname{claim}{claim}{claims}
\Crefname{claim}{Claim}{Claims}

\title{Follow the Regularized Leader Does Not Converge\\ in Constrained Optimization}

\author{%
  Ioannis Anagnostides\thanks{Ioannis Anagnostides and Nikolas Patris contributed equally.}\textsuperscript{\phantom{*},}\thanks{Carnegie Mellon University. \texttt{ianagnos@cs.cmu.edu}, \texttt{sandholm@cs.cmu.edu}}
  \and Ioannis Panageas\thanks{University of California, Irvine. \texttt{ipanagea@uci.edu}, \texttt{npatris@uci.edu}}
  \and Nikolas Patris\footnotemark[1]\textsuperscript{\phantom{*},}\footnotemark[3]
  \and Tuomas Sandholm\footnotemark[2]\textsuperscript{\phantom{\textdagger},}\thanks{Additional affiliations of Tuomas Sandholm: Strategy Robot, Inc.; Strategic Machine, Inc.; Optimized Markets, Inc.}
}

\date{}

\begin{document}

\pagenumbering{gobble}
\maketitle

\begin{abstract}
Follow the regularized leader (FTRL) is a foundational algorithm in online learning whose regret properties have been extensively studied for decades. However, in stark contrast to mirror descent---and gradient descent in particular---the convergence of FTRL to first-order stationary points in constrained nonconvex optimization was not well understood. In this paper, we show that continuous-time FTRL can fail to converge \emph{even asymptotically} in infinitely smooth constrained optimization problems: its Karush-Kuhn-Tucker (KKT) gap can remain bounded away from zero indefinitely. This is especially surprising in light of the fact that continuous-time FTRL guarantees a monotonic improvement of the objective value. As a result, we establish that non-convergent behavior of no-regret dynamics---which besets general variational inequality problems---can occur even in potential systems. From a technical standpoint, our construction relies on an infinitely differentiable function where the gradient flow dynamics exhibit decelerating oscillations without ever converging pointwise, which we in turn embed into higher-dimensional FTRL dynamics that sustain a perpetually large KKT gap. We complement our negative result by showing that FTRL with Euclidean regularization \emph{does} converge when the objective function is \emph{analytic}. This shows that the analytic property plays a more fundamental role for FTRL than for gradient flow: it governs not only pointwise convergence, but convergence to first-order stationarity itself.
\end{abstract}

\clearpage
\pagenumbering{arabic}
\section{Introduction}

\emph{Follow the regularized leader} ($\ftrl$) is a premier online optimization algorithm that has shaped modern machine learning~\citep{Kalai05:Efficient,Shalev-Shwartz12:Online,Hazan16:Introduction}. The regret properties of $\ftrl$---and variations thereof such as follow the perturbed leader~\citep{Kalai05:Efficient}---have engendered a vibrant area of research in the last two decades, and are by now well understood. On top of establishing minimax optimal regret bounds, $\ftrl$ provides a broad, unifying framework, encompassing celebrated algorithms such as multiplicative weights update~\citep{Littlestone94:Weighted,Arora12:Multiplicative}. As such, $\ftrl$ has driven empirical breakthroughs in complex domains, serving as the engine for large-scale click-through rate (CTR) prediction~\citep{McMahan13:Ad}, and being at the heart of AI milestones in solving zero-sum games~\citep{Moravvcik17:DeepStack,Brown17:Superhuman,Brown19:Superhuman,Farina21:Faster,Zhang26:General,Sokota26:Scalable}.

However, despite the tremendous amount of interest $\ftrl$ has received, its convergence properties as a first-order optimization algorithm have remained poorly understood. This is in sharp contrast to (online) mirror descent ($\md$)~\citep{Memirovskii83:Problem}, another seminal class of algorithms in online learning. Indeed, different incarnations of $\md$ algorithms are known to converge to first-order stationary points in time polynomial in $1/\epsilon$, where $\epsilon$ is the stationarity gap. For example, the convergence of gradient descent---which can be viewed as $\md$ endowed with Euclidean regularization---even in nonconvex problems is a cornerstone in optimization theory~\citep{Nesterov04:Introductory}. This raises the following fundamental question:
\begin{quote}
    \centering
    \emph{Does $\ftrl$ converge to first-order stationary points in constrained optimization?}
\end{quote}
For unconstrained problems, there is a well-known equivalence between $\ftrl$ and $\md$, which holds more generally even in the constrained setting under suitable assumptions on the regularizer (\emph{e.g.},~\citealp{McMahan11:Follow}). This implies that $\ftrl$ inherits the convergence of $\md$ in such settings. However, the general question beyond such special cases has remained wide open despite decades of research on the behavior of $\ftrl$. Recent work by~\citet{Anagnostides26:Doubly} established an exponential lower bound for $\ftrl$ in \emph{potential games}, which can be viewed as a special case of constrained (nonconvex) optimization problems. Potential games have played a central role in the development of game theory since the pathbreaking work of~\citet{Rosenthal73:Class}, and the convergence of $\ftrl$ to Nash equilibria---the natural counterpart of first-order stationary points in the case of games---has received considerable attention (we elaborate on related work in~\Cref{sec:related}). While that result exposes severe inefficiency of $\ftrl$, it leaves open the possibility of eventual convergence. 

In this paper, we take a significant step further: we show that $\ftrl$ can fail to converge \emph{even asymptotically} in constrained optimization problems, revealing a more fundamental separation between $\ftrl$ and $\md$.

\subsection{Our results}

In more detail, we examine the convergence of continuous-time $\ftrl$ in constrained optimization problems. We focus on the continuous-time version of $\ftrl$ because it always has the potential property (\Cref{lem:pot-prop}), in that it can never decrease the value of the potential function (which is to be maximized). As a result, one would expect that the stationarity gap of the dynamics would asymptotically be driven to zero; this would indeed be in line with prior results showing that $\mwu$---an incarnation of $\ftrl$---converges to Nash equilibria in (generic) potential games~\citep{Kleinberg09:Multiplicative,Palaiopanos17:Multiplicative}. On the other hand, the discrete-time version of $\ftrl$ only has the potential property when the learning rate is sufficiently small~\citep{Anagnostides26:Doubly}; indeed, $\ftrl$ with a large learning rate is known to exhibit a wide spectrum of non-convergent behavior, such as Li-Yorke chaos~\citep{Palaiopanos17:Multiplicative,Bielawski21:Follow}; those results do not apply in our setting because we examine the limiting case where the learning rate tends to zero.

As alluded to, the notion of convergence we consider is measured in terms of the \emph{Karush-Kuhn-Tucker (KKT)} gap (\Cref{def:KKTgap}), which is the standard notion of first-order suboptimality in optimization. This is weaker than pointwise convergence; as we shall see, dynamics such as gradient flow always converge in terms of the KKT gap but can still admit non-trivial limit sets.

In this context, our main result is summarized below.

\begin{theorem}
    \label{thm:main-informal}
    There exists an infinitely differentiable function $\Phi \in C^{\infty}$ on a convex and compact domain such that for any time $t$, the KKT gap of $\ftrl$ with Euclidean regularization on $\Phi$ is at least some constant bounded away from $0$. This holds even in constant dimensions.
\end{theorem}

This result contributes to the fundamental understanding of the convergence of learning algorithms---and $\ftrl$ in particular---in nonconvex constrained optimization problems, establishing that $\ftrl$ is inherently incapable of serving as a first-order optimizer.

Before we provide a detailed overview of our construction, we first highlight the intuition behind it. First-order stationarity is determined by the \emph{instantaneous} gradient, whereas $\ftrl$ chooses its strategy using the \emph{cumulative} gradient. As a result, for $\ftrl$ to act upon a positive gradient, it must first offset the negative gradients accumulated earlier. Our construction sustains this mismatch between instantaneous and cumulative gradient indefinitely, which explains how $\ftrl$ can keep improving the objective while failing to approach first-order stationarity. In contrast, projected gradient flow moves immediately along a feasible direction with positive instantaneous gradient.

\paragraph{Proof overview} The starting point of our construction concerns (continuous-time) gradient flow. Specifically, we rely on an instance in which gradient flow does not converge pointwise---but does converge in terms of the KKT gap. The fact that this phenomenon can occur is well known in the theory of dynamical systems~\citep{Palis12:Geometric}; to foreshadow our main positive result, we recall that---by virtue of the Łojasiewicz inequality~\citep{Lojasiewicz84:Trajectoires,Absil05:Convergence}---this can only happen under a non-analytic (but infinitely differentiable) function. In more detail, we first show that gradient flow follows the trajectory shown in~\Cref{fig:sinewave}, governed by $x_2 = \sin \left( \frac{\pi x_1}{1 - x_1} \right)$ (\Cref{lem:trapping}), thereby approaching the limit set $\{1\} \times [-1, 1]$. As time goes by, the dynamics in this instance move ever so slowly as the KKT gap approaches zero. The stalling of the dynamics notwithstanding, a critical property for the purpose of our construction is that the dynamics do not decelerate too quickly (\Cref{lem:x2_level_dwell_time}).

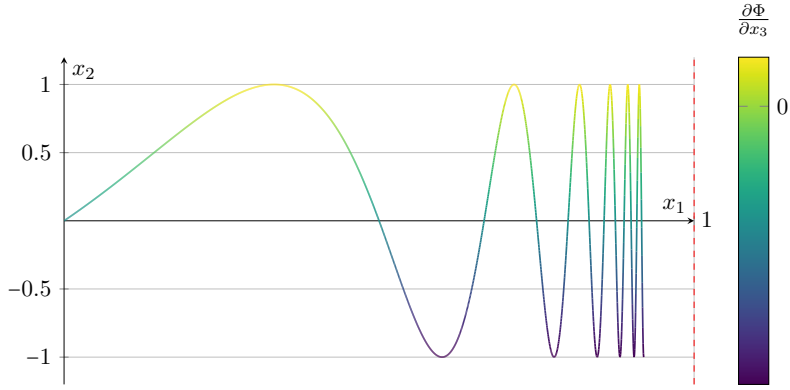
\begin{figure}
    \centering
    \scalebox{0.8}{\begin{tikzpicture}
    \begin{axis}[
        axis lines = middle,
        xlabel = {$x_1$},
        ylabel = {$x_2$},
        domain = 0:0.92,
        samples = 2000,
        width = 12cm,
        height = 7cm,
        grid = major,
        ymin = -1.2, ymax = 1.2,
        xmin = 0, xmax = 1,
        xtick = {1}, 
        xticklabel style={xshift=6pt, yshift=10pt}, 
        colorbar,
        colorbar style={
            title={$\frac{\partial \Phi}{\partial x_3} $}, 
            ytick={0},
            yticklabels={0}
        },
        colormap/viridis,
    ]
        \addplot [
            thick,
            mesh,
            point meta={sin(deg(pi*x/(1-x))) - 0.7},
        ] {sin(deg(pi*x/(1-x)))};
        \draw [red, dashed] (axis cs:1,-1.2) -- (axis cs:1,1.2);
    \end{axis}
\end{tikzpicture}}
    \caption{Schematic illustration of the trajectory of $\ftrl$ on the curve $x_2 = \sin\left( \frac{\pi x_1}{1 - x_1} \right)$. The color indicates the gradient observed on the $x_3$ axis, which is perpendicular to the page. Our argument shows that $\ftrl$ never acts upon the positive gradient on the $x_3$ axis.}
    \label{fig:sinewave}
\end{figure}

Our next step is to embed this instance into $\ftrl$ dynamics executed on a higher-dimensional space. The upshot is that we will force $\ftrl$ to always move on the $(x_1, x_2)$ plane, where the potential function is essentially flat when $x_1$ approaches $1$. At the same time, after a short transient phase, there will always be a direction perpendicular to the $(x_1, x_2)$ plane in which $\ftrl$ will observe a large KKT gap, but will still never act upon it. To realize this counterintuitive behavior, we maintain the invariance that the $x_1$ and $x_2$ coordinates of $\ftrl$ remain in the interior throughout, whereas the rest of the coordinates always remain inactive, in that they are stuck at zero. For the active coordinates, namely $(x_1, x_2)$, it is easy to show that $\ftrl$ is equivalent to gradient flow (\Cref{lemma:GF-FTRL}), matching the trajectory shown in~\Cref{fig:sinewave}. The more interesting claim is that at least one of the rest of the coordinates will observe a large KKT gap. The reason why this happens is because the \emph{cumulative} gradient observed by each inactive coordinate will be negative throughout, even though for a constant fraction of the time the instantaneous gradient is positive. This is where the oscillating pattern in the $(x_1, x_2)$ plane becomes critical: the key idea is to create a large negative gradient in some region of the space that neutralizes the positive gradient in aggregate (\Cref{fig:sinewave}).

To arrive at~\Cref{thm:main-informal}, we need to tie two loose ends. First, we need to guarantee that \emph{for any time} the KKT gap of $\ftrl$ is at least some constant bounded away from zero. To do so, we lift not just to three dimensions, as shown in~\Cref{fig:sinewave}, but to higher dimensions so that there is always a direction of improvement. Second, to argue about the basis of the induction---that movement only occurs on the $(x_1, x_2)$ plane---we incorporate an additional term in the objective value that induces large negative gradients in the beginning of the trajectory, but vanishes as $x_1 \to 1$. 

These steps are formalized and described in more detail in~\Cref{sec:gradient_flow,sec:ftrl}.

\paragraph{Convergence for analytic objectives} Finally, we show that the distinction between smooth and analytic objectives is fundamental for the convergence of $\ftrl$, establishing a positive counterpart to~\Cref{thm:main-informal}.

\begin{theorem}
    \label{thm:analytic-informal}
    For any analytic objective on a compact polytope, $\ftrl$ converges to a KKT point. In particular, these dynamics converge to a Nash equilibrium in finite potential games.
\end{theorem}

In a nutshell, the proof---discussed at greater length in~\Cref{sec:analytic}---combines the geometry of Euclidean projection onto a polytope with a tailored \L{}ojasiewicz gradient inequality. In fact, our argument establishes pointwise convergence, not just convergence to the set of KKT points.

\subsection{Further related work}
\label{sec:related}

We conclude this section by highlighting some key lines of research closely related to our paper.

\paragraph{Learning in potential games} The convergence of learning dynamics in potential games has been a fruitful area of research. Potential games can be viewed as the special case of constrained optimization where the potential function is multilinear (\Cref{sec:constrainedopt}). In that class of games, as discussed, first-order stationary points amount to Nash equilibria~\citep{Nash50:Equilibrium}. Many natural learning algorithms have been analyzed in potential games~\citep{Monderer96:Potential,Candogan13:Dynamics,Heliou17:Learning,Hart03:Regret,Panageas23:Exponential,Anagnostides26:Convergence,Palaiopanos17:Multiplicative,Kleinberg09:Multiplicative,Maheshwari25:Independent,Marden07:Regret,Zilberstein26:Distributed}. It is worth clarifying that these results concern specifically (last-)iterate convergence to Nash equilibria; a well-known fact is that the \emph{average} correlated distribution induced by any no-regret dynamics converges to the set of coarse correlated equilibria~\citep{Moulin78:Strategically}, which is a weaker solution concept---not least in terms of worst-case welfare~\citep{Blum08:Regret}. Our main result does not apply to finite potential games because there the potential function is a polynomial, and thus analytic---so gradient flow always converges pointwise~\citep{Lojasiewicz84:Trajectoires,Absil05:Convergence}. For Euclidean regularization, we establish convergence to a single KKT point for analytic objectives on compact polytopes (\Cref{thm:analytic-convergence}), and hence convergence to Nash equilibrium in every finite potential game (\Cref{cor:analytic-games}). A key reference point is the result of~\citet{Kleinberg09:Multiplicative}, which showed that continuous-time multiplicative weights update---the special case of $\ftrl$ under entropic regularization---converges to the set of \emph{weakly stable} Nash equilibria in congestion games from almost every initialization; for almost every choice of congestion costs, all such equilibria are pure. Also, \citet{Anagnostides26:Doubly} showed that a lazy version of $\ftrl$---in which a step is only taken if it is bound to improve the potential by a sufficient amount---always converges when the potential function is smooth. Lazy updates preclude our counterexample as the dynamics would eventually stop moving in the direction where the potential function is essentially flat.

\paragraph{FTRL and sparsity} Beyond regret guarantees, $\ftrl$ offers an advantage relative to $\md$ in promoting sparsity~\citep{Xiao10:Dual,McMahan11:Follow}. For example, this practical benefit has been demonstrated in large-scale click-through rate prediction~\citep{McMahan13:Ad}. Intuitively, $\ftrl$ can result in sparser solutions because it allows coordinates to remain inactive despite occasional gradients favoring activation. Our result exposes a limitation of this persistence in constrained nonconvex optimization.

\paragraph{Complexity of constrained optimization} From a complexity standpoint, computing a first-order stationary point of a smooth function is complete for the class $\CLS$, which stands for continuous local search~\citep{Daskalakis11:Continuous,Fearnley22:Complexity}. This result holds even in two dimensions and provides strong evidence that no polynomial-time algorithm exists. Related hardness results have been established for computing Nash equilibria in potential games~\citep{Babichenko20:Communication,Babichenko21:Settling}. Indeed, while projected gradient descent converges in such problems, the number of iterations need not be polynomially bounded in the input encoding length. In contrast, our result concerning $\ftrl$ applies even when both the number of dimensions and the stationarity gap are absolute constants.

\paragraph{Nonconvergence in games} From a broader standpoint, our main result shows that, surprisingly, natural no-regret dynamics can exhibit nonconvergent oscillations even when there is an underlying potential. It is well known that learning dynamics can fail to converge in general-sum games; for example, \citet{Daskalakis10:Learning} gave an example where multiplicative weights update can fail to converge to Nash equilibria, and there is a rich literature on this subject (\emph{e.g.},~\citealp{Mertikopoulos18:Cycles,Vlatakis20:No}). Our work significantly refines those prior results by reconciling such counterexamples with the potential property, albeit applying beyond finite games. On the other hand, it is worth pointing out that FTRL has been shown to converge in more structured classes of problems, such as bargaining games~\citep{Kamp25:Last}.

\paragraph{Forgetfulness and speed of convergence} Finally, \citet{Cai24:Fast} established a separation between \emph{optimistic} variants of $\ftrl$ and $\md$ in zero-sum games. Their main finding, that fast convergence requires forgetfulness, mirrors our main result. We take a step further by showing that $\ftrl$ does not merely suffer from slow convergence, but can fail to converge altogether. Also, our result concerns a fundamentally different class of problems and algorithms compared to~\citet{Cai24:Fast}. In fact, \citet{Orabona26:Last} recently showed that optimistic multiplicative weights \emph{does} converge in convex-concave problems.
\section{Preliminaries}

Before we dive into our analysis, in this section, we provide the necessary background on constrained optimization problems (\Cref{sec:constrainedopt}) and follow the regularized leader (\Cref{sec:prels-ftrl}).

\subsection{Constrained optimization and KKT gap}
\label{sec:constrainedopt}

We begin by introducing the underlying class of problems in which we are interested. We let $\cX$ be a convex and compact subset of a Euclidean space, and let $\Phi : \cX \to \R$ be a function to be maximized. $\Phi$ will generally be nonconcave, making global optimization intractable. Instead, the main goal will be to compute a \emph{first-order stationary point} of $\Phi$, in a sense formally specified below (\Cref{def:KKTgap}). $\Phi$ is assumed to be \emph{smooth}, meaning that there exists a parameter $L > 0$ such that $\| \nabla \Phi(\vx) - \nabla \Phi(\vx') \|_2 \leq L \|\vx - \vx' \|_2$ for any $\vx, \vx' \in \cX$. In fact, as we shall see, our lower bound applies even when $\Phi \in C^\infty$, which is to say that the $k$th derivatives of $\Phi$ are continuous for any $k \in \N$.

We measure local suboptimality in terms of the \emph{Karush-Kuhn-Tucker (KKT)} gap~\citep{Boyd04:Convex}. There are multiple equivalent ways (up to polynomial factors) of defining the KKT gap; in this paper, we make use of the following definition.

\begin{definition}[KKT gap]
    \label{def:KKTgap}
    For a function $\Phi \in C^1$ defined on a convex and compact domain $\cX$, the \emph{KKT gap} of a point $\vx \in \cX$ is defined as
    \begin{equation}
        \label{eq:KKTgap}
        \KKTgap(\vx) \defeq \max_{\vx' \in \cX} \langle \vx' - \vx, \nabla \Phi(\vx) \rangle.
    \end{equation}
\end{definition}

An equivalent definition is in terms of the fixed-point gap of the gradient ascent operator $\vx \mapsto \Pi_{\cX}( \vx + \eta \nabla \Phi(\vx))$, where $\Pi_{\cX}(\cdot)$ denotes the Euclidean projection~\citep{Facchinei03:Finite}.

We will be interested in the case where the constraint set $\cX$ can be decomposed as the Cartesian product $\cX \defeq \cX_1 \times \dots \times \cX_n$. In that case, the KKT gap per~\eqref{eq:KKTgap} of a point $\vx = (\vx_1, \dots, \vx_n)$ can be equivalently expressed as $\sum_{i=1}^n \max_{\vx'_i \in \cX_i} \langle \vx_i' - \vx_i, \nabla_{\vx_i} \Phi(\vx) \rangle$. In particular, the KKT gap will be large if there is a considerable first-order deviation benefit \emph{for at least} a single component $i \in [n]$; our lower bound relies on this simple fact. 

As discussed, a special case of computing KKT points in constrained optimization is identifying Nash equilibria in potential games~\citep{Rosenthal73:Class,Monderer96:Potential}. Specifically, in $n$-player potential games, $\cX$ can be decomposed as a Cartesian product of $n$ probability simplices, and $\Phi$ is a multilinear function with respect to the strategies of each player. 

\subsection{Follow the regularized leader and the potential property}
\label{sec:prels-ftrl}

We proceed by introducing follow the regularized leader ($\ftrl$)~\citep{Hazan16:Introduction}. We focus on the continuous-time version of $\ftrl$, which can be thought of as the limit case of discrete-time $\ftrl$ when the learning rate approaches zero. To put this into context, it is worth reiterating that $\ftrl$ with a large learning rate is known to exhibit non-convergent behavior~\citep{Palaiopanos17:Multiplicative,Bielawski21:Follow}; our result is fundamentally different as it is not an artifact of a poorly tuned learning rate. $\ftrl$ was historically developed as a regularized version of \emph{fictitious play}~\citep{Robinson51:iterative}, the foundational algorithm in learning in games. Fictitious play can be viewed as the extreme version of $\ftrl$ when the learning rate is arbitrarily large; unlike $\ftrl$, fictitious play does not guarantee the no-regret property.

In this context, continuous-time $\ftrl$ with Euclidean regularization can be defined as
\begin{equation} 
\label{eq:ftrl}
    \vx(t) = \argmax_{\vx' \in \cX} \left\{ \left\langle \vx', \int_0^t \nabla_{\vx} \Phi(\vx(\tau)) d\tau \right\rangle - \frac{1}{2} \|\vx' \|_2^2 \right\},
\end{equation}
where $\Phi$ is the differentiable objective function; in particular, if $\vec{0} \in \cX$, \eqref{eq:ftrl} prescribes initializing at $\vec{0}$. Equivalently, \eqref{eq:ftrl} can be expressed as
\begin{equation}
    \label{eq:ftrl-proj}
    \vx(t) = \argmin_{\vx' \in \cX}  \left\| \vx' - \int_{0}^t \nabla_{\vx} \Phi(\vx(\tau)) d \tau \right \|_2 = \Pi_{\cX} \left( \int_{0}^t \nabla_{\vx} \Phi(\vx(\tau)) d \tau \right).
\end{equation}

With an arbitrary initialization $\vy_0$, the update is $\vx(t)=\Pi_{\cX}(\vy_0+\int_0^t\nabla\Phi(\vx(\tau))\,d\tau)$. Unless stated otherwise, we assume zero initialization.

A key property of continuous-time $\ftrl$ is that it can never decrease the value of $\Phi$. In other words, as one would expect, $\Phi$ serves as a potential under the $\ftrl$ dynamics; for the discrete-time version of this property when the learning rate is small enough, we refer to~\citet{Anagnostides26:Doubly}. 

\begin{restatable}[Potential property for $\ftrl$]{lemma}{potpropftrl}
    \label{lem:pot-prop}
    For $\ftrl$ on $\Phi \in C^1$, the objective value $\Phi(\vx(t))$ is nondecreasing, and for almost every $t\geq0$,
    \[
        \frac{d}{dt} \Phi(\vx(t)) \ge 0.
    \]
\end{restatable}
We provide the simple proof of this property in~\Cref{appendix:proofs}. As a result, \Cref{lem:pot-prop} reassures us that the potential value is monotonically nondecreasing over time.

\paragraph{Gradient flow} Before we proceed, it is worth comparing $\ftrl$ with the (projected) gradient flow dynamics, whereby 
\[
\dot{\vx}(t) = \Pi_{T_\cX(\vx(t))} (\nabla_{\vx} \Phi (\vx(t)) ),
\]
where $T_\cX(\vx)$ denotes the tangent cone of $\cX$ at $\vx$. When the maximum in~\eqref{eq:ftrl} is consistently attained in the interior of $\cX$, the first-order optimality conditions imply that $\ftrl$ coincides with (unconstrained) gradient flow; we will use this simple observation in our counterexample (\Cref{lemma:GF-FTRL}). Gradient flow also enjoys the potential property, as it guarantees $\frac{d}{dt} \Phi(\vx(t)) = \langle \dot{\vx}, \nabla_{\vx} \Phi(\vx(t)) \rangle \geq \| \dot{\vx}(t) \|_2^2$ (by the chain rule and standard properties of the tangent cone projection). Integrating yields $\int_{0}^\infty \| \dot{\vx}(t)  \|_2^2 dt \leq O(1)$ since $\Phi$ is bounded. Importantly, the norm of the velocity $\| \dot{\vx}(t) \|_2$ under (constrained) gradient flow satisfies $\KKTgap(\vx(t))\leq\operatorname{diam}(\cX)\|\dot{\vx}(t)\|_2$ almost everywhere, implying first-order stationarity as the dynamics go along. As we shall see, this is not so for $\ftrl$.

\section{Gradient flow without pointwise convergence}
\label{sec:gradient_flow}

The starting point of our construction is a two-dimensional instance in which gradient flow does not converge pointwise. The fact that this phenomenon can occur is well known in dynamical systems~\citep{Palis12:Geometric}, but the point of this section is to establish some key properties that we rely on when we turn to $\ftrl$ on a higher-dimensional space (\Cref{sec:ftrl}).

We begin by defining the potential. It is designed so that the
trajectory approaches (from the left) the boundary \(x_1=1\) while perpetually oscillating in the
\(x_2\)-coordinate. For \(x_1<1\), we set
\[
    u=\frac{1}{1-x_1},
    \quad
    v=\pi(u-1),
    \quad
    \delta=x_2-\sin v.
\]
Here, \(u\) measures the proximity to the boundary \(x_1=1\), \(v\) is the total phase
of the oscillation, and \(\delta\) is the signed deviation from the oscillatory
curve \(x_2=\sin v\). The goal is to construct a potential whose gradient flow maintains the invariance $\delta(t) = 0$. To this end, we let
\[
    a(u)
    \defeq
    \frac{u^4e^{-u}}
    {1+\pi^2u^4\cos^2 v}
    \text{ and }
    L(u)
    \defeq
    -\pi a(u)\cos v,
\]
and we define
\(\Phi:[0,1]\times[-1,1]\to\mathbb R\) by
\begin{equation}
    \label{eq:pot-flow}
    \Phi(x_1,x_2)
    \defeq
    -
    e^{-u}
    -
    L(u)\delta
    -
    \frac12 e^{-u}\delta^2,
    \quad x_1<1.
\end{equation}
At the boundary \(x_1=1\), we define \(\Phi(1,x_2)\) by the limit
\(x_1\to1^{-}\) of the right-hand side of \eqref{eq:pot-flow}, which is
\(0\). The resulting potential is smooth up to the boundary
(\Cref{lem:phi_smooth}). For the enlarged constraint set used in \Cref{sec:ftrl}, we extend the same formula to all $x_1<1$ and $x_2\in\R$, and set $\Phi(x_1,x_2)=0$ for $x_1 \geq 1$; this extension is $C^\infty$ by the same argument.
The role of the linear term \(L(u)\delta\) in~\eqref{eq:pot-flow} is to make the oscillatory curve
\(
    \mathcal M=\{\delta=0\}
\)
forward invariant under the gradient flow. 
As we shall show, along this curve, the dynamics reduce to the scalar differential equation
\(
    \dot u=a(u)>0,
\)
so \(u(t)\) increases monotonically and the trajectory moves through successive
oscillations while approaching the boundary \(x_1=1\).

In the remainder of this section, we first formalize that the oscillatory curve $\{ \delta = 0 \}$ is forward invariant (\Cref{lem:trapping}). The second and more involved step involves obtaining dwell-time estimates for the motion through each oscillation (\Cref{lem:x2_level_dwell_time}). All proofs are deferred to~\Cref{appendix:proofs}.

\subsection{Invariant oscillatory curve}

We first show that the potential function constructed in~\eqref{eq:pot-flow} guarantees the following key invariance.

\begin{restatable}[Forward invariance]{lemma}{forinv}
\label{lem:trapping}
Let $\mathcal M
\defeq
\{(x_1,x_2):\delta=0\}
=
\left\{
(x_1,x_2):
x_2=\sin\left(\frac{\pi x_1}{1-x_1}\right)
\right\}$. Then $\mathcal{M}$ is forward invariant under the gradient flow on $\Phi$.
\end{restatable}

We state two immediate corollaries stemming from~\Cref{lem:trapping}.

\begin{restatable}{corollary}{initcurve}
\label{cor:initialization_on_invariant_curve}
If gradient flow is initialized at $(x_1(0), x_2(0)) = (0, 0)$, then $\delta(t) = 0$ for all $t \geq 0$. Moreover, $\dot u(t) = a(u) > 0$, so \(u(t)\) is strictly increasing.
\end{restatable}

\begin{restatable}[Limit sets]{corollary}{omegalimit}
\label{lem:omega_limit_invariant_curve}
If gradient flow is initialized on \(\mathcal M\), then $u(t)\to+\infty$ and $x_1(t) \to 1$. Moreover, the $\omega$-limit set reads $\omega(\vx(0))=\{1\}\times[-1,1]$.
\end{restatable}

\subsection{Dwell-time bounds}
\label{subsec:dwell_time}

Next, we estimate how much time the gradient flow spends in different parts of each oscillation. These dwell-time bounds will be crucial when we embed the
two-dimensional gradient flow instance into the higher-dimensional \(\ftrl\)
dynamics, as they allow us to control the cumulative gradient observed by each
auxiliary coordinate. Referring back to \Cref{fig:sinewave}, the key requirement
is that, within each oscillation, the positive cumulative gradient is subsumed
by the negative cumulative gradient. This requires the flow to spend sufficiently
little time in the positive-gradient region; this subsection sets the groundwork for establishing that property.

\paragraph{Phase windows} We first introduce the phase windows. We assume that the gradient flow is
initialized on the invariant curve \(\mathcal M\). 
During the \(k\)th oscillation, with \(k\ge 2\), we write
\[
v=2\pi(k-1)+\phi,
\quad
\phi\in[0,2\pi],
\quad
U_k\defeq 2k-1.
\]
In words, \(\phi\) represents the local phase, while \(v\) is the absolute phase. For \(m\ge 1\), we let 
\[
\ell^{(m)} \defeq \frac{\pi}{2^m},
\quad
h^{(m)} \defeq \frac{\ell^{(m)}}{2}=\frac{\pi}{2^{m+1}},
\quad
N^{(m)} \defeq \frac{2\pi}{h^{(m)}}=2^{m+2}.
\]
Here, \(\ell^{(m)}\) is the window length and \(h^{(m)}\) is the shift between consecutive windows, so adjacent windows overlap by \(h^{(m)}\). Furthermore, for \(i \in \{0,\ldots,N^{(m)}-2\}\), we define the non-wrapping windows $I^{(m)}_i
\defeq [ih^{(m)},ih^{(m)}+\ell^{(m)}]
\subset[0,2\pi]$. The final window crosses \(2\pi\), and is written as
\(
    I^{(m)}_{N^{(m)}-1}
    \defeq
    [2\pi-h^{(m)},2\pi]\cup[0,h^{(m)}],
\)
where the wrapping is understood modulo \(2\pi\), so both pieces lie within the same oscillation. \Cref{fig:vis} contains a visualization of the phase windows.

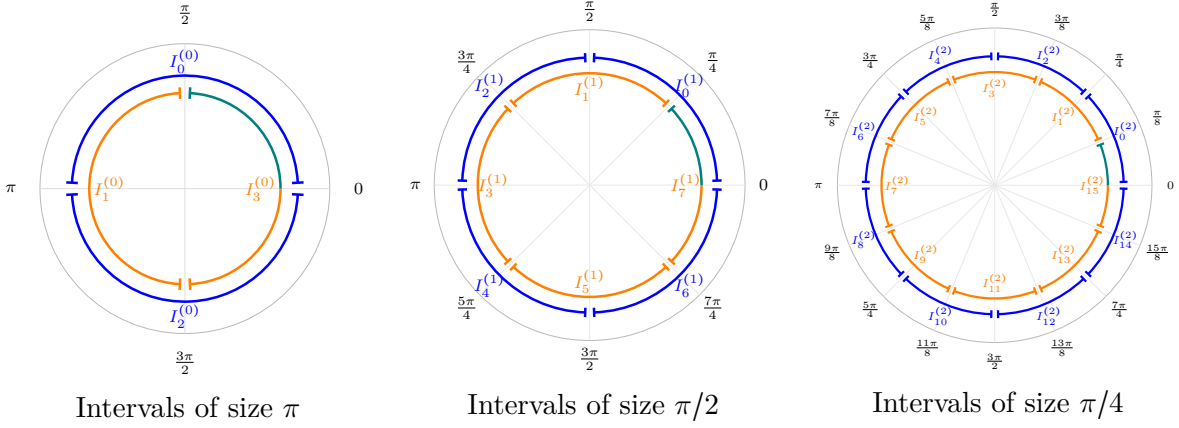
\begin{figure}[t]
    \centering

    \begin{minipage}{0.32\textwidth}
        \centering
        \resizebox{\linewidth}{!}{\begin{tikzpicture}[scale=2.4]
    \def\gap{3}

    \def\rblue{0.78}
    \def\rorange{0.66}

    \draw[thin, gray!55] (0,0) circle (1);

    \foreach \ang/\lab in {
        0/{0},
        90/{\frac{\pi}{2}},
        180/{\pi},
        270/{\frac{3\pi}{2}}
    } {
        \draw[thin, gray!25]
            (0,0) -- ({cos(\ang)}, {sin(\ang)});

        \node[font=\scriptsize]
            at ({1.20*cos(\ang)}, {1.20*sin(\ang)})
            {$\lab$};
    }

    \draw[very thick, blue]
        ({\rblue*cos(0+\gap)}, {\rblue*sin(0+\gap)})
        arc[start angle={0+\gap}, end angle={180-\gap}, radius=\rblue];

    \draw[very thick, blue]
        ({(\rblue-0.04)*cos(0+\gap)}, {(\rblue-0.04)*sin(0+\gap)})
        -- ({(\rblue+0.04)*cos(0+\gap)}, {(\rblue+0.04)*sin(0+\gap)});
    \draw[very thick, blue]
        ({(\rblue-0.04)*cos(180-\gap)}, {(\rblue-0.04)*sin(180-\gap)})
        -- ({(\rblue+0.04)*cos(180-\gap)}, {(\rblue+0.04)*sin(180-\gap)});

    \node[font=\scriptsize, blue]
        at ({0.88*cos(90)}, {0.88*sin(90)})
        {$I^{(0)}_0$};

    \draw[very thick, orange]
        ({\rorange*cos(90+\gap)}, {\rorange*sin(90+\gap)})
        arc[start angle={90+\gap}, end angle={270-\gap}, radius=\rorange];

    \draw[very thick, orange]
        ({(\rorange-0.04)*cos(90+\gap)}, {(\rorange-0.04)*sin(90+\gap)})
        -- ({(\rorange+0.04)*cos(90+\gap)}, {(\rorange+0.04)*sin(90+\gap)});
    \draw[very thick, orange]
        ({(\rorange-0.04)*cos(270-\gap)}, {(\rorange-0.04)*sin(270-\gap)})
        -- ({(\rorange+0.04)*cos(270-\gap)}, {(\rorange+0.04)*sin(270-\gap)});

    \node[font=\scriptsize, orange]
        at ({0.52*cos(180)}, {0.52*sin(180)})
        {$I^{(0)}_1$};

    \draw[very thick, blue]
        ({\rblue*cos(180+\gap)}, {\rblue*sin(180+\gap)})
        arc[start angle={180+\gap}, end angle={360-\gap}, radius=\rblue];

    \draw[very thick, blue]
        ({(\rblue-0.04)*cos(180+\gap)}, {(\rblue-0.04)*sin(180+\gap)})
        -- ({(\rblue+0.04)*cos(180+\gap)}, {(\rblue+0.04)*sin(180+\gap)});
    \draw[very thick, blue]
        ({(\rblue-0.04)*cos(360-\gap)}, {(\rblue-0.04)*sin(360-\gap)})
        -- ({(\rblue+0.04)*cos(360-\gap)}, {(\rblue+0.04)*sin(360-\gap)});

    \node[font=\scriptsize, blue]
        at ({0.88*cos(270)}, {0.88*sin(270)})
        {$I^{(0)}_2$};


    \draw[very thick, orange]
        ({\rorange*cos(270+\gap)}, {\rorange*sin(270+\gap)})
        arc[start angle={270+\gap}, end angle={360}, radius=\rorange];

    \draw[very thick, teal]
        ({\rorange*cos(0)}, {\rorange*sin(0)})
        arc[start angle={0}, end angle={90-\gap}, radius=\rorange];

    \draw[very thick, orange]
        ({(\rorange-0.04)*cos(270+\gap)}, {(\rorange-0.04)*sin(270+\gap)})
        -- ({(\rorange+0.04)*cos(270+\gap)}, {(\rorange+0.04)*sin(270+\gap)});

    \draw[very thick, teal]
        ({(\rorange-0.04)*cos(90-\gap)}, {(\rorange-0.04)*sin(90-\gap)})
        -- ({(\rorange+0.04)*cos(90-\gap)}, {(\rorange+0.04)*sin(90-\gap)});

    \node[font=\scriptsize, orange]
        at ({0.74*cos(315)}, {0.56*sin(0)})
        {$I^{(0)}_3$};
\end{tikzpicture}}
        \vspace{0.5em}
        {Intervals of size $\pi$}
    \end{minipage}
    \hfill
    \begin{minipage}{0.32\textwidth}
        \centering
        \resizebox{\linewidth}{!}{\begin{tikzpicture}[scale=2.6]
    \def\gap{2}

    \def\rblue{0.82}
    \def\rorange{0.72}

    \def\step{45}
    \def\len{90}

    \draw[thin, gray!55] (0,0) circle (1);

    \foreach \i/\lab in {
        0/{0},
        1/{\frac{\pi}{4}},
        2/{\frac{\pi}{2}},
        3/{\frac{3\pi}{4}},
        4/{\pi},
        5/{\frac{5\pi}{4}},
        6/{\frac{3\pi}{2}},
        7/{\frac{7\pi}{4}}
    } {
        \pgfmathsetmacro{\ang}{\i*\step}

        \draw[thin, gray!18]
            (0,0) -- ({cos(\ang)}, {sin(\ang)});

        \node[font=\scriptsize]
            at ({1.12*cos(\ang)}, {1.12*sin(\ang)})
            {$\lab$};
    }

    \foreach \i in {0,...,6} {
        \pgfmathsetmacro{\start}{\i*\step}
        \pgfmathsetmacro{\finish}{\start+\len}
        \pgfmathsetmacro{\mid}{\start+\len/2}
        \pgfmathtruncatemacro{\parity}{mod(\i,2)}

        \ifnum\parity=0
            \def\col{blue}
            \def\r{\rblue}
            \def\rlab{0.91}
        \else
            \def\col{orange}
            \def\r{\rorange}
            \def\rlab{0.62}
        \fi

        \draw[very thick, \col]
            ({\r*cos(\start+\gap)}, {\r*sin(\start+\gap)})
            arc[start angle={\start+\gap}, end angle={\finish-\gap}, radius=\r];

        \draw[very thick, \col]
            ({(\r-0.03)*cos(\start+\gap)}, {(\r-0.03)*sin(\start+\gap)})
            -- ({(\r+0.03)*cos(\start+\gap)}, {(\r+0.03)*sin(\start+\gap)});

        \draw[very thick, \col]
            ({(\r-0.03)*cos(\finish-\gap)}, {(\r-0.03)*sin(\finish-\gap)})
            -- ({(\r+0.03)*cos(\finish-\gap)}, {(\r+0.03)*sin(\finish-\gap)});

        \node[font=\scriptsize, \col]
            at ({\rlab*cos(\mid)}, {\rlab*sin(\mid)})
            {$I^{(1)}_{\i}$};
    }


    \draw[very thick, orange]
        ({\rorange*cos(7*\step+\gap)}, {\rorange*sin(7*\step+\gap)})
        arc[start angle={7*\step+\gap}, end angle={360}, radius=\rorange];

    \draw[very thick, teal]
        ({\rorange*cos(0)}, {\rorange*sin(0)})
        arc[start angle={0}, end angle={\step-\gap}, radius=\rorange];

    \draw[very thick, orange]
        ({(\rorange-0.03)*cos(7*\step+\gap)}, {(\rorange-0.03)*sin(7*\step+\gap)})
        -- ({(\rorange+0.03)*cos(7*\step+\gap)}, {(\rorange+0.03)*sin(7*\step+\gap)});

    \draw[very thick, teal]
        ({(\rorange-0.03)*cos(\step-\gap)}, {(\rorange-0.03)*sin(\step-\gap)})
        -- ({(\rorange+0.03)*cos(\step-\gap)}, {(\rorange+0.03)*sin(\step-\gap)});

    \node[font=\scriptsize, orange]
        at ({0.62*cos(7*\step+\len/2)}, {0.62*sin(7*\step+\len/2)})
        {$I^{(1)}_{7}$};
\end{tikzpicture}}
        \vspace{0.5em}
        {Intervals of size $\pi/2$}
    \end{minipage}
    \hfill
    \begin{minipage}{0.32\textwidth}
        \centering
        \resizebox{\linewidth}{!}{\begin{tikzpicture}[scale=2.9]
    \def\gap{1.5}

    \def\rblue{0.82}
    \def\rorange{0.72}

    \def\step{22.5}
    \def\len{45}

    \draw[thin, gray!55] (0,0) circle (1);

    \foreach \i/\lab in {
        0/{0},
        1/{\frac{\pi}{8}},
        2/{\frac{\pi}{4}},
        3/{\frac{3\pi}{8}},
        4/{\frac{\pi}{2}},
        5/{\frac{5\pi}{8}},
        6/{\frac{3\pi}{4}},
        7/{\frac{7\pi}{8}},
        8/{\pi},
        9/{\frac{9\pi}{8}},
        10/{\frac{5\pi}{4}},
        11/{\frac{11\pi}{8}},
        12/{\frac{3\pi}{2}},
        13/{\frac{13\pi}{8}},
        14/{\frac{7\pi}{4}},
        15/{\frac{15\pi}{8}}
    } {
        \pgfmathsetmacro{\ang}{\i*\step}

        \draw[thin, gray!18]
            (0,0) -- ({cos(\ang)}, {sin(\ang)});

        \node[font=\tiny]
            at ({1.12*cos(\ang)}, {1.12*sin(\ang)})
            {$\lab$};
    }

    \foreach \i in {0,...,14} {
        \pgfmathsetmacro{\start}{\i*\step}
        \pgfmathsetmacro{\finish}{\start+\len}
        \pgfmathsetmacro{\mid}{\start+\len/2}
        \pgfmathtruncatemacro{\parity}{mod(\i,2)}

        \ifnum\parity=0
            \def\col{blue}
            \def\r{\rblue}
            \def\rlab{0.90}
        \else
            \def\col{orange}
            \def\r{\rorange}
            \def\rlab{0.62}
        \fi

        \draw[very thick, \col]
            ({\r*cos(\start+\gap)}, {\r*sin(\start+\gap)})
            arc[start angle={\start+\gap}, end angle={\finish-\gap}, radius=\r];

        \draw[very thick, \col]
            ({(\r-0.025)*cos(\start+\gap)}, {(\r-0.025)*sin(\start+\gap)})
            -- ({(\r+0.025)*cos(\start+\gap)}, {(\r+0.025)*sin(\start+\gap)});

        \draw[very thick, \col]
            ({(\r-0.025)*cos(\finish-\gap)}, {(\r-0.025)*sin(\finish-\gap)})
            -- ({(\r+0.025)*cos(\finish-\gap)}, {(\r+0.025)*sin(\finish-\gap)});

        \node[font=\tiny, \col]
            at ({\rlab*cos(\mid)}, {\rlab*sin(\mid)})
            {$I^{(2)}_{\i}$};
    }


    \draw[very thick, orange]
        ({\rorange*cos(15*\step+\gap)}, {\rorange*sin(15*\step+\gap)})
        arc[start angle={15*\step+\gap}, end angle={360}, radius=\rorange];

    \draw[very thick, teal]
        ({\rorange*cos(0)}, {\rorange*sin(0)})
        arc[start angle={0}, end angle={\step-\gap}, radius=\rorange];

    \draw[very thick, orange]
        ({(\rorange-0.025)*cos(15*\step+\gap)}, {(\rorange-0.025)*sin(15*\step+\gap)})
        -- ({(\rorange+0.025)*cos(15*\step+\gap)}, {(\rorange+0.025)*sin(15*\step+\gap)});

    \draw[very thick, teal]
        ({(\rorange-0.025)*cos(\step-\gap)}, {(\rorange-0.025)*sin(\step-\gap)})
        -- ({(\rorange+0.025)*cos(\step-\gap)}, {(\rorange+0.025)*sin(\step-\gap)});

    \node[font=\tiny, orange]
        at ({0.62*cos(15*\step+\len/2)}, {0.62*sin(15*\step+\len/2)})
        {$I^{(2)}_{15}$};
\end{tikzpicture}}
        \vspace{0.5em}
        {Intervals of size $\pi/4$}
    \end{minipage}

    \caption{
    Overlapping phase windows of decreasing size. The part shown in \textcolor{teal}{teal} is the final wrapped part of the last phase window, represented modulo \(2\pi\) within the same oscillation.}
    \label{fig:vis}
\end{figure}

\paragraph{Dwell times} For the construction of the potential in \Cref{sec:ftrl}, we express the dwell time in terms of level sets of the \(x_2\)-coordinate rather than phase windows. To this end, for each phase window \(I_i^{(m)}\), \(i \in \{0,\ldots,N^{(m)}-1\}\), we define the associated \(x_2\)-level set by
\[
A_i^{(m)}
\defeq
\sin\bigl(I_i^{(m)}\bigr)
\defeq
\left\{
x_2 \in[-1,1]:
x_2 = \sin\phi
\text{ for some }\phi\in I_i^{(m)}
\right\},
\]
which is the corresponding interval in the \(x_2\)-coordinate. Then, during the \(k\)th oscillation, we let \(T^{(m)}_{k,i}\) denote the amount of time for which
\(
x_2(t)\in A_i^{(m)},
\)
and let \(Q^{(m)}_{k,i}\) denote the amount of time for which
\(
x_2(t)\notin A_i^{(m)}.
\)

\begin{restatable}[Dwell-time bounds in terms of \(x_2\)-level sets]{lemma}{dwellxtwo}
\label{lem:x2_level_dwell_time}
Suppose that the gradient flow is initialized on the invariant curve
\(\mathcal M\). 
Then, for every \(k\ge2\), \(m\ge5\), and
\(i \in \{ 0,\ldots,N^{(m)}-1 \} \),
\begin{equation}
    \label{eq:dwellbound}
    T^{(m)}_{k,i}
    \leq
    \frac{16}{2^{m}}
    Q^{(m)}_{k,i}.
\end{equation}
\end{restatable}

For the construction in \Cref{sec:ftrl}, the time \(T^{(m)}_{k,i}\) spent in \(A_i^{(m)}\) bounds the time during which the corresponding auxiliary gradient is positive; throughout the complementary time \(Q^{(m)}_{k,i}\), this gradient is strictly negative. When \(m\) is large, \eqref{eq:dwellbound} shows that \(T^{(m)}_{k,i}\) is a small fraction of the oscillation's duration.

\section{Main result: from gradient flow to FTRL}
\label{sec:ftrl}

To obtain our main result, we now embed the gradient flow dynamics of~\Cref{sec:gradient_flow} into a higher-dimensional space. Specifically, the first two
coordinates \((x_1,x_2)\) are the original variables of gradient flow, while the remaining
coordinates $(x_3, \dots, x_{N^{(m)}+2})$ serve an auxiliary role. The key invariance we will establish is that $x_i(t) = 0$ for any $i \in \{3,\ldots,N^{(m)}+2\}$ and $t \geq 0$, so $\ftrl$ in this higher-dimensional space will remain trapped in the $(x_1, x_2)$ plane, mimicking the trajectory of gradient flow. In accordance with the dwell-time bounds of~\Cref{lem:x2_level_dwell_time}, we associate $x_{3+i}$ with the $i$th $x_2$-level interval. The constraint set is defined as the product of intervals $\cX_1 \times \cX_2 \times [0, 1] \times \dots \times [0, 1]$, where $\cX_1=[-1,2]$ and $\cX_2=[-2,2]$; this guarantees that constrained gradient flow with respect to $(x_1, x_2)$ coincides with the unconstrained gradient flow by virtue of~\Cref{lem:trapping}.

\newcommand{\boxB}{\widetilde{A}}
\newcommand{\boxalpha}{\widetilde{\alpha}^{(m)}}
\newcommand{\boxbeta}{\widetilde{\beta}^{(m)}}

\paragraph{Higher-dimensional potential} We now introduce the augmented potential function. We recall from~\cref{subsec:dwell_time} that each phase window \(I^{(m)}_i\) determines the interval
\[
A^{(m)}_i
=
[\alpha_i^{(m)},\beta_i^{(m)}],
\quad
\text{where }
\alpha_i^{(m)} \defeq \min_{\phi\in I_i^{(m)}}\sin\phi
, \ 
\beta_i^{(m)} \defeq \max_{\phi\in I_i^{(m)}}\sin\phi
.
\]
In particular, for the wrapping window, we have \(A_{N^{(m)}-1}^{(m)} = \sin\left([2\pi-h^{(m)},2\pi]\cup[0,h^{(m)}]\right) = [-\sin h^{(m)},\sin h^{(m)}]\). 
%
For each \(A_i^{(m)}\), write
\(\Delta_i^{(m)}\defeq\beta_i^{(m)}-\alpha_i^{(m)}>0\)
for its length. We introduce a corresponding trimmed interval
\(\boxB_i^{(m)}\defeq[\boxalpha_i,\boxbeta_i]\),
whose role will become clear shortly.
If \(A_i^{(m)}\) lies strictly inside \([-1,1]\), we remove
one quarter of its length from each endpoint, so that
\[
\boxB_i^{(m)}
\defeq 
\left[ \alpha_i^{(m)}+\frac{\Delta_i^{(m)}}{4}, \beta_i^{(m)}-\frac{\Delta_i^{(m)}}{4}
\right],
\]
and hence $|\boxB_i^{(m)}| = |A_i^{(m)}|/ 2$. The only intervals that need separate treatment are the endpoint intervals, which come in two types: those reaching $1$ and those reaching $-1$. For $m \geq 1$, no interval reaches both endpoints because each phase window has length at most \(\pi/2\). For the endpoint intervals, we modify the construction slightly. At the endpoint \(-1\) or \(1\), we extend the interval by one quarter of its original length, while trimming the opposite side by one quarter. This extension ensures that the quadratic \(q_i^{(m)}\) defined below is strictly positive at the corresponding endpoint \(x_2=\pm1\), which is needed for a uniformly positive KKT gap.
That is, if \(\beta_i^{(m)}=1\), we set $\boxB_i^{(m)}
    \defeq
    \Big[
    \alpha_i^{(m)}+\Delta_i^{(m)}/4,
    \,
    \beta_i^{(m)}+\Delta_i^{(m)}/4
    \Big]$, while if \(\alpha_i^{(m)}=-1\), we set $\boxB_i^{(m)}
    \defeq
    \Big[
    \alpha_i^{(m)}-\Delta_i^{(m)}/4,
    \,
    \beta_i^{(m)}-\Delta_i^{(m)}/4
    \Big]$; in either case, $ |\boxB_i^{(m)}| = |A_i^{(m)}|$.

We associate with the interval \(\boxB_i^{(m)}\) the quadratic term $q_i^{(m)}(x_2)
\defeq
(
    x_2-\boxalpha_i
)
(
    \boxbeta_i-x_2
)$, so that the initial potential function $\Phi$ per~\eqref{eq:pot-flow} in~\Cref{sec:gradient_flow} is now recast as
\begin{equation}
    \label{eq:interm-pot}
    \widetilde{\Phi} (x_1,x_2,x_3,\ldots,x_{N^{(m)}+2}) \defeq \Phi(x_1,x_2)
    +
    \sum_{i=0}^{N^{(m)}-1}
    x_{3+i}\,q_i^{(m)}(x_2).
\end{equation}
%
As a result, the gradient of~\eqref{eq:interm-pot} with respect to \(x_{3+i}\)
is precisely \(q_i^{(m)}(x_2)\).
We can now explain the role of the trimmed intervals:
they allow us to explicitly bound the positive values of
\(q_i^{(m)}\) inside \(A_i^{(m)}\) and obtain a uniform strictly
negative upper bound on \([-1,1]\setminus A_i^{(m)}\).
Combined with the dwell-time estimate, these bounds will show
that its integral over a complete oscillation is negative. At the same time, as long as all auxiliary coordinates remain zero, the added terms do not change the \((x_1,x_2)\)-trajectory since the added contribution to $\partial \widetilde{\Phi} /\partial x_2$ is always multiplied by one of the auxiliary variables. Furthermore, the intervals constructed above are such that at any point in time at least one coordinate will observe a positive gradient (\Cref{lem:uniform_kkt_gap}), so $\ftrl$ will experience a large KKT gap.

Now, under the constraint set $\cX_1 \times \cX_2 \times [0, 1] \times \dots \times [0, 1]$, $\ftrl$ can be decomposed as $x_i(t) = \Pi_{\cX_i}( \int_{0}^t \nabla_{x_i} \widetilde{\Phi}(\vx(\tau)) d \tau )$ for each coordinate $i$ (\Cref{cor:decomposition}). The first direct consequence of this decomposition is that maintaining the invariance $x_{3+i} = 0$ reduces to showing that the observed \emph{cumulative} gradient is nonpositive---the projection of any nonpositive number into $[0, 1]$ is $0$.

\begin{claim}
    \label{claim:negative-cum}
    For $i\in \{ 3, \ldots, N^{(m)}+2 \}$, if $\int_{0}^t \nabla_{x_i} \widetilde{\Phi}(\vx(\tau)) d \tau \leq 0$ for any $t \geq 0$, then $x_{i}(t) = 0$.
\end{claim}

\paragraph{Bounding the cumulative gradients} We now employ the dwell-time estimate from \cref{lem:x2_level_dwell_time} to control the
cumulative gradient observed by each auxiliary coordinate. We fix an oscillation
\(k\) and an index \(i\in\{0,\ldots,N^{(m)}-1\}\). We recall that \(\Delta_i^{(m)}=\beta_i^{(m)}-\alpha_i^{(m)}>0\)
is the length of \(A_i^{(m)}\). The quadratic function
\(q_i^{(m)}(x_2)\) satisfies the following bounds. First, if \(A_i^{(m)}\) contains neither endpoint \(1\) nor \(-1\), then
\[
    \max_{x_2 \in \boxB_i^{(m)}}q_i^{(m)}(x_2)
    =
    \frac{(\Delta_i^{(m)})^2}{16}
    \text{ and }
    q_i^{(m)}(\alpha_i^{(m)})
    =
    q_i^{(m)}(\beta_i^{(m)})
    =
    -\frac{3(\Delta_i^{(m)})^2}{16}.
\]
Hence, by concavity,
\(
    q_i^{(m)}(x_2)
    \le
    \nicefrac{-3(\Delta_i^{(m)})^2}{16}
\)
whenever \(x_2\notin A_i^{(m)}\). For the endpoint intervals, 
\[
    \max_{x_2 \in \boxB_i^{(m)}}q_i^{(m)}(x_2)
    =
    \frac{(\Delta_i^{(m)})^2}{4},
\]
and concavity gives
\(
    q_i^{(m)}(x_2)
    \le
    \nicefrac{- 5(\Delta_i^{(m)})^2}{16}
\)
whenever \(x_2\in[-1,1]\setminus A_i^{(m)}\). 
Let $t_k$ denote the time at which $u(t_k)=2k-1$, so that $t_1=0$ and the $k$th oscillation corresponds to $[t_k,t_{k+1}]$. Since \(q_i^{(m)}\) attains its global maximum
at the midpoint of \(\boxB_i^{(m)}\), the positive upper
bounds above also hold throughout \(A_i^{(m)}\).
We therefore use the corresponding positive bound over
the time \(T_{k,i}^{(m)}\) spent inside \(A_i^{(m)}\)
and the negative bound over the complementary time
\(Q_{k,i}^{(m)}\), obtaining
\[
\begin{aligned}
    \int_{t_k}^{t_{k+1}} q_i^{(m)}(x_2(t))\,dt
    &\leq
    \max_{x_2\in\boxB_i^{(m)}}q_i^{(m)}(x_2)\,T_{k,i}^{(m)}
    +\sup_{x_2\in[-1,1]\setminus A_i^{(m)}}q_i^{(m)}(x_2)\,Q_{k,i}^{(m)} \\
    &\leq
    \frac{(\Delta_i^{(m)})^2}{16}
    \begin{cases}
        T_{k,i}^{(m)}-3Q_{k,i}^{(m)}
        &\text{if } A_i^{(m)}\subset(-1,1),\\
        4T_{k,i}^{(m)}-5Q_{k,i}^{(m)}
        &\text{otherwise}.
    \end{cases}
\end{aligned}
\]
By \cref{lem:x2_level_dwell_time}, both bounds are negative for $k \geq 2$ and $m\geq5$, since $T_{k,i}^{(m)} \leq 16 \, 2^{-m} Q_{k,i}^{(m)}$ and $16 \, 2^{-m} < 1$. Thus, over each complete oscillation, the positive contribution is outweighed by the negative contribution. We summarize this property in the lemma below.

\begin{lemma}
    \label{lemma:cum-grad}
    Suppose that $(x_1(t), x_2(t))$ evolves as the (unconstrained) gradient flow executed on~\eqref{eq:pot-flow}. If $m \geq 5$ in~\eqref{eq:interm-pot}, then $\int_{t_k}^{t_{k+1}} q_i^{(m)}(x_2(t))\,dt<0$ for every $k\geq2$ and $i\in\{0,\ldots,N^{(m)}-1\}$.
\end{lemma}

\paragraph{Initialization} To complete the construction, we need to handle the initialization. To do so, we incorporate an additional term in~\eqref{eq:interm-pot}, namely
\begin{equation}
    \label{eq:newpot}
\widetilde{\Phi}' (x_1,x_2,x_3,\ldots,x_{N^{(m)}+2}) \defeq \widetilde{\Phi} (x_1,x_2,x_3,\ldots,x_{N^{(m)}+2}) - C \sum_{i=0}^{N^{(m)} - 1} x_{3 + i} \psi(x_1),
\end{equation}
where $C > 0$ is a sufficiently large constant and $\psi(x_1)=e^{-1/(1-x_1)}$ for $x_1<1$, while $\psi(x_1)=0$ for $x_1\geq1$; this extension is $C^\infty$. The role of the last term in~\eqref{eq:newpot} is this: when $C$ is large enough, each coordinate $x_{3+i}$ will observe a large negative gradient in the beginning (so it will remain at $0$), but as $x_1$ approaches $1$ that term will quickly vanish and become negligible. As before, this additional term will not affect the gradient observed by $x_1$ or $x_2$ as long as all auxiliary coordinates remain zero.

Choosing $C$ large enough makes the cumulative gradients nonpositive during the first two oscillations. For later oscillations, the negative contribution accumulated during the preceding complete oscillation outweighs the positive contribution during any part of the current one, provided that $m\geq8$. This leads to the following lemma, whose proof is given in~\Cref{appendix:proofs}.

\begin{restatable}{lemma}{auxcumgradient}
\label{lem:auxiliary-cumulative-gradients}
    Suppose that $(x_1(t), x_2(t))$ evolves as the (unconstrained) gradient flow executed on~\eqref{eq:pot-flow}. If $m\geq8$ and $C > 0$ is large enough, $\int_{0}^t \nabla_{x_i} \widetilde{\Phi}'(\vx(\tau)) d \tau \leq 0$ for any $ i\in\{3,\ldots,N^{(m)}+2\}$ and $t \geq 0$.
\end{restatable}
For the remainder of the construction, we fix $m=8$. Combining this lemma with~\Cref{claim:negative-cum} shows that $x_i(t) = 0$ for any $t \geq 0$ and $ i \in \{3,\ldots,N^{(m)}+2\}$. What remains to argue is that on the $(x_1, x_2)$ plane, $\ftrl$ executed on~\eqref{eq:newpot} adheres to the motion described in~\Cref{lem:trapping}. This follows from the equivalence of $\ftrl$ and gradient flow when the latter remains in the interior of the constraint set, as we point out below; the proof follows by simply integrating.

\begin{claim}[Equivalence of gradient flow to $\ftrl$ in the interior]
    \label{lemma:GF-FTRL}
    Consider gradient flow on a function $\Phi$ and constraint set $\cX = \cX_1 \times \dots \times \cX_n$ such that for each $i \in S \subseteq [n]$, $x_i(t)$ is always in the interior of $\cX_i$. For each coordinate $i \in S$, $x_i(t) = \int_{0}^t \nabla_{x_i} \Phi(\vx(\tau)) d \tau = \Pi_{\cX_i} ( \int_{0}^t \nabla_{x_i} \Phi(\vx(\tau)) d \tau ) $.
\end{claim}

To conclude, we start from the oscillatory gradient-flow trajectory from~\Cref{lem:trapping}, augmented with auxiliary coordinates fixed at zero. This higher-dimensional trajectory is precisely $\ftrl$: \Cref{lem:auxiliary-cumulative-gradients} together with~\Cref{claim:negative-cum} shows this for the auxiliary coordinates, while \Cref{lemma:GF-FTRL} shows it for the first two coordinates. As a result, there is a constant $\gamma_m > 0$ such that $\max_i q_i^{(m)}(x_2)\geq\gamma_m$ for every $x_2\in[-1,1]$ (\Cref{lem:uniform_kkt_gap}). We choose $T$ such that $C\psi(x_1(t))\leq\gamma_m/2$ for all $t \geq T$. For the final objective $\widetilde{\Phi}'$, an auxiliary coordinate then gives a KKT gap of at least $\gamma_m/2$. On $[0,T]$, the first two coordinates have a nonzero gradient and remain in the interior, again implying a strictly positive lower bound for the KKT gap. The smaller of these two constants is a uniform lower bound for all $t \geq 0$. Together, these observations establish~\Cref{thm:main-informal}.

\section{Convergence for analytic objectives}
\label{sec:analytic}

Our counterexample relies on a smooth objective that is \emph{not} analytic. In this section, we establish that this distinction is fundamental: $\ftrl$ run on an analytic objective defined on a polyhedral set is guaranteed to converge to a KKT point.

\begin{restatable}{theorem}{analyticconvergence}
    \label{thm:analytic-convergence}
    Let $\cX \subset \R^d$ be a nonempty compact polytope, and let $\Phi$ be real analytic on an open neighborhood of $\cX$. Continuous-time $\ftrl$ with Euclidean regularization has finite trajectory length, $\int_0^\infty \|\dot{\vx}(t)\|_2\,dt < \infty$, and thereby converges to a point $\vx^\star \in \cX$ with $\KKTgap(\vx^\star)=0$.
\end{restatable}

In particular, any finite potential game admits a \emph{polynomial} (in fact, multilinear) potential defined on products of simplices, giving the following consequence.

\begin{restatable}{corollary}{analyticgames}
    \label{cor:analytic-games}
    In every finite potential game, continuous-time $\ftrl$ with Euclidean regularization and arbitrary initialization converges to a Nash equilibrium.
\end{restatable}

The detailed proof of~\Cref{thm:analytic-convergence} is deferred to~\Cref{appendix:analytic}, but here we provide an overview. The first key claim, established through some basic projection geometry, is that the velocity satisfies almost everywhere
\[
    \|\dot{\vx}(t)\|_2\geq\|\nabla_{F(\vx(t))}\Phi(\vx(t))\|_2,
\]
where $F(\vx)$ is the \emph{smallest} face of $\cX$ containing $\vx$ and $\nabla_F\Phi$ is the gradient restricted to its affine hull (\Cref{lem:analytic-speed}). Second, the \L{}ojasiewicz gradient inequality can be leveraged to establish a lower bound uniformly across all the faces of $\cX$ (\Cref{lem:analytic-uniform}). Namely, writing $E(t)=\ell-\Phi(\vx(t))$, where $\ell$ is the objective value in the limit, we obtain $\|\dot{\vx}(t)\|_2\geq cE(t)^\theta$ almost everywhere for some $c>0$ and $\theta\in[1/2,1)$, assuming $t$ is large enough and $E(t)>0$; crucially, this holds even when the active face keeps changing. Combining it with the fact that $\dot E=-\|\dot{\vx}\|_2^2$ almost everywhere and integrating implies finite length and thereby convergence to a KKT point $\vx^\star$.
\section{Conclusions}

In conclusion, we showed that follow the regularized leader ($\ftrl$) can fail to converge even asymptotically to first-order stationary points in constrained optimization. This establishes a strong separation between the convergence of $\ftrl$ and mirror descent ($\md$), and shows that lack of convergence of no-regret dynamics can manifest itself even in potential systems. On the positive side, we showed that the analytic property restores convergence for $\ftrl$.

Although we focused on Euclidean regularization, we expect the construction to readily extend to broader classes of regularizers that maintain coordinates fixed at the boundary when their cumulative gradients are negative enough. For steep regularizers such as negative entropy, however, the strategy will always be an interior point for any finite time, so extending our construction would require additional ideas.


\section*{Acknowledgments}

Ioannis Panageas was supported by NSF grant CCF-2454115. Tuomas Sandholm is supported by NIH award A240108S001, the Vannevar Bush Faculty Fellowship ONR N00014-23-1-2876, and National Science Foundation grant RI-2312342. Ioannis Anagnostides is grateful to Emanuel Tewolde and Brian Hu Zhang for helpful discussions. We also thank anonymous NeurIPS reviewers for many helpful comments.

\paragraph{AI disclosure} The first version of this paper was prepared without AI assistance. In subsequent revisions, we used various versions of GPT to simplify the arguments, mainly in the proof of our positive convergence result, and identify relatively minor technical issues.

\bibliographystyle{plainnat}
\bibliography{dairefs}

\clearpage
\appendix

\section{Omitted proofs}
\label{appendix:proofs}

This section contains the omitted proofs. We begin by showing that continuous-time $\ftrl$ has the potential property.

\potpropftrl*

\begin{proof}
Let $\vy(t)=\int_0^t\nabla\Phi(\vx(\tau))\,d\tau$, so that $\dot{\vy}(t)=\nabla\Phi(\vx(t))$ and $\vx(t)=\Pi_{\cX}(\vy(t))$. Both trajectories are Lipschitz, and hence differentiable almost everywhere. Firm nonexpansiveness of the Euclidean projection implies
\[
    \langle\vy(t+h)-\vy(t),\vx(t+h)-\vx(t)\rangle
    \geq\|\vx(t+h)-\vx(t)\|_2^2.
\]
Dividing by $h^2$ and taking the limit $h\to0$ yields
\[
    \frac{d}{dt}\Phi(\vx(t))
    =\langle\dot{\vy}(t),\dot{\vx}(t)\rangle
    \geq\|\dot{\vx}(t)\|_2^2\geq0.
\]
This concludes the proof.
\end{proof}

We next point out a simple decomposition property of $\ftrl$ when the constraint set is a Cartesian product; the proof is immediate. 

\begin{corollary}
    \label{cor:decomposition}
    If $\cX = \cX_1 \times \cX_2 \times \dots \times \cX_n$, then $\ftrl$ can be equivalently expressed as
    \[
    \vx_i(t) = \argmax_{\vx_i' \in \cX_i} \left\{ \left\langle \vx_i', \int_0^t \nabla_{\vx_i} \Phi(\vx(\tau)) d\tau \right\rangle - \frac{1}{2} \|\vx_i' \|_2^2 \right\} \quad \forall i \in [n].
    \]
\end{corollary}



\begin{lemma}[Smoothness of \(\Phi\)]
\label{lem:phi_smooth}
\(\Phi\) defined in~\eqref{eq:pot-flow} is \(C^\infty\) up to the boundary \(x_1=1\).
\end{lemma}

\begin{proof}
For \(x_1<1\), the function \(\Phi\) is smooth. Indeed, it is obtained by composition, addition, multiplication, and division of smooth functions, and the denominator in \(L(u)\), $1+\pi^2u^4\cos^2 v$, is strictly positive. Thus, the only issue is smoothness at the boundary
\(x_1=1\), which corresponds to \(u\to+\infty\).

To argue about this, we will use the standard fact that exponential decay dominates polynomial
growth: for every \(M>0\), $u^M e^{-u}\to 0$ as $u \to +\infty$. This implies that any expression of the form $e^{-u} R(u,\sin v,\cos v)$, where \(R\) grows at most polynomially in \(u\), tends to zero as \(u\to+\infty\), since \(\sin v\) and \(\cos v\) are bounded. For the expressions appearing below, this decay persists under differentiation. Each expression is a finite sum of terms $e^{-u}P(u,\sin v,\cos v)/(1+\pi^2u^4\cos^2v)^j$, with $P$ a polynomial and $j \geq0 $ an integer. The operator $\partial_{x_1}=u^2\partial_u$ preserves this form, so every derivative with respect to $x_1$ is still bounded by an exponentially decaying factor times a polynomial in $u$, and therefore tends to zero.

We now apply this observation to the three terms in \eqref{eq:pot-flow}, which we recall here for completeness:
\[
\Phi=-e^{-u}-L(u)\delta-\frac12 e^{-u}\delta^2.
\]
First, \(-e^{-u}\) is flat at \(x_1=1\). That is, all of its \(x_1\)-derivatives
are of the form $e^{-u}P(u)$, where \(P\) is a polynomial. Thus, all of them tend to zero as
\(u\to+\infty\). Next, we consider the linear term \(-L(u)\delta\). Since
\[
    L(u)
    =
    -
    \frac{\pi u^4e^{-u}\cos v}
    {1+\pi^2u^4\cos^2 v},
\]
and the denominator is bounded below by \(1\), the factor \(L(u)\) contains the exponentially decaying factor \(e^{-u}\). Therefore \(L(u)\), together with all
of its derivatives with respect to $x_1$, tends to zero as \(u\to+\infty\). Moreover, $\delta=x_2-\sin v$. Differentiating \(\delta\) with respect to \(x_1\) only produces polynomial
powers of \(u\) and bounded trigonometric factors. So, by the product rule, every \(x_1\)-derivative of \(L(u)\delta\) is a finite sum of terms of the form $e^{-u}R(u,\sin v,\cos v,x_2)$, where \(R\) grows at most polynomially in \(u\). Thus, all such derivatives tend to zero as \(u\to+\infty\).  Finally, we consider the quadratic term $-\frac12 e^{-u}\delta^2$. The factor \(e^{-u}\) is flat at the boundary. Differentiating \(\delta^2\) with respect to \(x_1\) produces only polynomial powers of \(u\), powers of \(\delta\), and bounded trigonometric factors. Since smoothness is a local property in the \(x_2\)-variable, we may restrict \(x_2\) to a compact set; on such a set, \(\delta=x_2-\sin v\) is bounded. Therefore, every \(x_1\)-derivative of \(e^{-u}\delta^2\) also tends to zero as \(u\to+\infty\).

It remains to consider differentiation with respect to \(x_2\). Since \(\Phi\) depends at most quadratically on \(x_2\), these derivatives can be computed immediately and read as follows.
\[
\partial_{x_2}\Phi
=
-L(u)-e^{-u}\delta,
\quad
\partial_{x_2}^2\Phi
=
-e^{-u},
\quad
\partial_{x_2}^k\Phi
=
0
\quad\text{for all }
k\geq 3.
\]
Now, each of these expressions, and each of their derivatives with respect to $x_1$, tends to zero
as \(u\to+\infty\) by the preceding argument. Therefore, every mixed partial derivative $\partial_{x_1}^\alpha\partial_{x_2}^\beta \Phi$ has a continuous extension to \(x_1=1\), and its boundary value is zero. Thus, the extension defined by $\Phi(1,x_2)=0$ is \(C^\infty\) up to the boundary \(x_1=1\).
\end{proof}

\forinv*

\begin{proof}
We work with the gradient-ascent dynamics \(\dot{\vx}=\nabla\Phi(\vx)\), where
\[
    \Phi
    =
    -e^{-u}
    -
    L(u)\delta
    -
    \frac12 e^{-u}\delta^2,
    \quad
    \delta=x_2-\sin v,
    \quad
    v=\pi(u-1).
\]
To prove that \(\mathcal M=\{\delta=0\}\) is forward invariant, it is enough to show that \(\dot\delta=0\) whenever \(\delta=0\). We first compute the two velocity components on \(\mathcal M\). Since
\(\partial_{x_2}\delta=1\),
\begin{equation}
    \label{eq:trapping_dot_x2}
    \dot x_2
    =
    \partial_{x_2}\Phi
    =
    -L(u)-e^{-u}\delta
    =
    -L(u)
\end{equation}
on $\mathcal{M}$. For the \(x_1\)-component, we use the fact that
\(\partial u/\partial x_1=u^2\), \(\partial v/\partial x_1=\pi u^2\), and thus
\[
    \frac{\partial\delta}{\partial x_1}
    =
    -\pi u^2\cos v.
\]
On \(\mathcal M\), all terms in \(\partial_{x_1}\Phi\) containing an explicit
factor of \(\delta\) vanish, so
\[
    \partial_{x_1}\Phi
    =
    u^2e^{-u}
    -
    L(u)\frac{\partial\delta}{\partial x_1}
    =
    u^2e^{-u}
    +
    \pi u^2L(u)\cos v .
\]
Since \(\dot x_1=\partial_{x_1}\Phi\) and \(\dot u=u^2\dot x_1\), we get
\begin{equation}
    \label{eq:trapping_dot_u_before_L}
    \dot u
    =
    u^4e^{-u}
    +
    \pi u^4L(u)\cos v .
\end{equation}
Substituting \(L(u)=-\pi a(u)\cos v\),
\[
    \dot u
    =
    u^4e^{-u}
    -
    \pi^2u^4a(u)\cos^2 v .
\]
By the definition of
\[
    a(u)
    =
    \frac{u^4e^{-u}}
    {1+\pi^2u^4\cos^2 v},
\]
we have \(u^4e^{-u}=a(u)(1+\pi^2u^4\cos^2 v)\). Therefore, we conclude that $\dot u=a(u)>0$ on $\mathcal{M}$. It remains to check that the tracking error is preserved. Differentiating
\(\delta=x_2-\sin v\) along the flow gives
\begin{equation}
    \label{eq:trapping_dot_delta}
    \dot\delta
    =
    \dot x_2-\cos v\,\dot v
    =
    \dot x_2-\pi\cos v\,\dot u .
\end{equation}
Combining with~\eqref{eq:trapping_dot_x2}, we obtain
\[
    \dot\delta
    =
    -L(u)-\pi a(u)\cos v
\]
on $\mathcal{M}$. This is zero because \(L(u)=-\pi a(u)\cos v\). Thus, \(\dot\delta=0\) whenever \(\delta=0\), so \(\mathcal M\) is forward invariant, and \(\dot u=a(u)>0\) along
\(\mathcal M\).
\end{proof}

\omegalimit*

\begin{proof}
By \cref{cor:initialization_on_invariant_curve}, if the trajectory is initialized on \(\mathcal M\), then it remains on \(\mathcal M\) for all future times, and \(\dot u>0\) holds. Moreover, $0<a(u)\leq u^4e^{-u}$ is bounded for $u\geq1$, so $u(t)$ cannot diverge in finite time. We now show that \(u(t)\) cannot converge to a finite value. Suppose, for the sake of contradiction, that $u(t)\to u_\ast<+\infty$. Since \(u(t)\) is increasing, this would mean that the trajectory approaches a finite point on the invariant curve. But \(\dot u\) is continuous and strictly positive at \(u_\ast\), by definition. Therefore, there exists \(\eta>0\) such that, for all sufficiently large \(t\), $\dot u(t)\ge \eta>0$. This is incompatible with the fact that \(u(t)\) converges to the finite value \(u_\ast\), thereby implying $u(t)\to+\infty$. Since \(x_1(t)=1-1/u(t)\), we immediately obtain $x_1(t) \to 1$. It remains to identify the accumulation set in the \(x_2\)-coordinate. Since \(u(t)\to+\infty\), the phase $v(t)=\pi(u(t)-1)$ also tends to \(+\infty\). Moreover, \(v(t)\) is continuous and strictly increasing because \(\dot v(t)=\pi\dot u(t)>0\), so the phase passes through infinitely many full periods. Along the invariant curve, $x_2(t)=\sin(v(t))$, so \(x_2(t)\) has every value in \([-1,1]\) as an accumulation point.

In particular, let us fix any \(y\in[-1,1]\), and choose
\(\alpha\in[-\pi/2,\pi/2]\) such that \(\sin\alpha=y\). Since \(v(t)\) is continuous, strictly increasing, and tends to \(+\infty\), for every sufficiently large integer \(n\) there is a time \(t_n\) such that $v(t_n)=\alpha+2\pi n$. Then, $x_2(t_n)=\sin(v(t_n))=y$ while \(x_1(t_n)\to1\). Therefore,
\[
    (x_1(t_n),x_2(t_n))\to(1,y).
\]
Since \(y\in[-1,1]\) was arbitrary, every point of
\(\{1\}\times[-1,1]\) is an accumulation point. Conversely, every accumulation point must lie in \(\{1\}\times[-1,1]\), because
\(x_1(t)\to1\) and \(x_2(t)\in[-1,1]\). Therefore, $\omega(\vx(0))=\{1\}\times[-1,1]$, as claimed.
\end{proof}

\dwellxtwo*

\begin{proof}
By~\Cref{cor:initialization_on_invariant_curve}, the phase satisfies
\[
    \dot\phi=\dot v=\pi\dot u=\pi a(u)
    =\frac{\pi u^4e^{-u}}{1+\pi^2u^4\cos^2 v}>0.
\]
Thus, we can compute dwell times by changing variables from time to phase.
During the \(k\)th oscillation, \(u=U_k+\phi/\pi\) and \(\cos v=\cos\phi\), so the time density is
\begin{equation}
    \label{eq:phase_time_density}
    g_k(\phi)\defeq\frac{dt}{d\phi}
    =\frac{1}{\dot v}
    =\frac{e^u}{\pi u^4}+\pi e^u\cos^2\phi.
\end{equation}
In particular, \(g_k(\phi)\,d\phi\) is the time spent while the local phase moves through the infinitesimal interval \(d\phi\).
On the invariant curve, \(x_2(t)=\sin\phi(t)\) during each oscillation. For each $i$, we let
\[
    S_i\defeq\{\phi\in[0,2\pi]:\sin\phi\in A_i^{(m)}\}.
\]
Since $\sin\phi=\sin(\pi-\phi)$, this set is the union of $I_i^{(m)}$ and its reflection $\pi-I_i^{(m)}$, interpreted modulo $2\pi$. The union has measure at most $2\ell^{(m)}$, including for wrapping windows and windows containing extrema. Thus,
\[
    T_{k,i}^{(m)}=\int_{S_i}g_k(\phi)\,d\phi
    \quad\text{and}\quad
    Q_{k,i}^{(m)}=\int_{[0,2\pi]\setminus S_i}g_k(\phi)\,d\phi.
\]
Since $k\geq2$, we have $U_k\geq3$, so the formula for the density gives
\[
    g_k(\phi)\leq e^{U_k}e^2\left(\frac{1}{81\pi}+\pi\right)
    \quad\text{and}\quad
    g_k(\phi)\geq e^{U_k}\pi\cos^2\phi.
\]
As a result,
\[
    T_{k,i}^{(m)}
    \leq e^{U_k}e^2\left(\frac{1}{81\pi}+\pi\right)2\ell^{(m)},
\]
while
\[
    \begin{aligned}
        Q_{k,i}^{(m)}
        &\geq e^{U_k}\pi\left(\int_0^{2\pi}\cos^2\phi\,d\phi
            -\int_{S_i}\cos^2\phi\,d\phi\right)\\
        &\geq e^{U_k}\pi(\pi-2\ell^{(m)}).
    \end{aligned}
\]
Using $\ell^{(m)}=\pi2^{-m}$ and $m\geq5$, we obtain
\[
    \frac{T_{k,i}^{(m)}}{Q_{k,i}^{(m)}}
    \leq\frac{2e^2(1+1/(81\pi^2))}{1-2^{1-m}}\,2^{-m}
    <16\,2^{-m}.
\]
This proves the claimed bound for every $i \in \{0,\ldots,N^{(m)}-1\}$, with both dwell times measured during the same oscillation.
\end{proof}

\auxcumgradient*
\begin{proof}
Fix $m\geq8$ and $i\in\{0,\ldots,N^{(m)}-1\}$. We set $r\defeq16\,2^{-m}$ and let
\[
    R_i(t)\defeq\int_0^t
    \bigl(q_i^{(m)}(x_2(\tau))-C\psi(x_1(\tau))\bigr)\,d\tau
\]
denote the cumulative gradient of coordinate $x_{3+i}$. Since $\psi(x_1(t))=e^{-u(t)}$ along the trajectory, choosing
\[
    M\defeq\max_{\substack{0\leq j<N^{(m)}\\z\in[-1,1]}}
        \max\{q_j^{(m)}(z),0\}
    \quad\text{and}\quad
    C\geq e^5(M+1)
\]
ensures that, on the first two oscillations,
\[
    q_i^{(m)}(x_2(t))-C\psi(x_1(t))
    \leq M-Ce^{-5}\leq-1,
    \quad 0\leq t\leq t_3,
\]
because $1\leq u(t)\leq5$ on this interval. In particular, $R_i(t)\leq0$ for every $t\in[0,t_3]$. The bounds preceding \Cref{lemma:cum-grad}, together with \cref{lem:x2_level_dwell_time}, imply that, for every $k\geq2$,
\[
    \int_{t_k}^{t_{k+1}}q_i^{(m)}(x_2(t))\,dt
    \leq\frac{(\Delta_i^{(m)})^2}{16}
        \bigl(4T_{k,i}^{(m)}-3Q_{k,i}^{(m)}\bigr)
    \leq-\frac{(\Delta_i^{(m)})^2}{16}(3-4r)Q_{k,i}^{(m)}.
\]
Moreover, the density formula~\eqref{eq:phase_time_density} and $U_{k+1}=U_k+2$ give $g_{k+1}(\phi)\leq e^2g_k(\phi)$. Since the phase set corresponding to $A_i^{(m)}$ is the same in each oscillation, it follows that
\[
    T_{k+1,i}^{(m)}\leq e^2T_{k,i}^{(m)}.
\]
Along the trajectory, $q_i^{(m)}(x_2)$ can be positive only when $x_2\in A_i^{(m)}$, and is bounded above by $(\Delta_i^{(m)})^2/4$. Thus, for any part of the $k$th oscillation,
\[
    \int_{t_k}^t q_i^{(m)}(x_2(\tau))\,d\tau
    \leq\frac{(\Delta_i^{(m)})^2}{4}T_{k,i}^{(m)},
    \quad t\in[t_k,t_{k+1}].
\]
We now argue by induction. Suppose that $k\geq3$ and $R_i(t)\leq0$ for all $t\leq t_k$. Since the correction term is nonpositive and $R_i(t_{k-1})\leq0$, the bound for the preceding complete oscillation gives
\[
    R_i(t_k)\leq-\frac{(\Delta_i^{(m)})^2}{16}(3-4r)Q_{k-1,i}^{(m)}.
\]
As a result, for $t\in[t_k,t_{k+1}]$,
\[
    \begin{aligned}
        R_i(t)
        &\leq-\frac{(\Delta_i^{(m)})^2}{16}(3-4r)Q_{k-1,i}^{(m)}
            +\frac{(\Delta_i^{(m)})^2}{4}T_{k,i}^{(m)}\\
        &\leq\frac{(\Delta_i^{(m)})^2}{16}
            \bigl(-(3-4r)+4e^2r\bigr)Q_{k-1,i}^{(m)}
        \leq0,
    \end{aligned}
\]
where the second inequality uses $T_{k,i}^{(m)}\leq e^2T_{k-1,i}^{(m)}\leq e^2rQ_{k-1,i}^{(m)}$. The last inequality holds because $m\geq8$ implies $2^m>\frac{64}{3}(1+e^2)$. The base case holds on $[0,t_3]$, so induction proves the claim for all $t\geq0$.
\end{proof}

\begin{restatable}[Coverage and overlap of the trimmed intervals]{lemma}{CovXTwo}
\label{lem:trimmed_boxes_cover}
For every \(m\ge1\), the trimmed intervals
\(\boxB_i^{(m)}\) cover \([-1,1]\), and consecutive
intervals overlap by at least
\[
    \omega^{(m)}
    \defeq
    2\sin^3\left(\frac{h^{(m)}}{2}\right)
    \sin\left(\frac{3h^{(m)}}{2}\right)>0.
\]
That is, with indices interpreted modulo \(N^{(m)}\),
\[
    \left|
        \boxB_i^{(m)}\cap \boxB_{i+1}^{(m)}
    \right|
    \ge
    \omega^{(m)}
    \quad
    \text{for every } i \in \{0,\ldots,N^{(m)}-1\}.
\]
\end{restatable}

\begin{proof}
Since the phase windows \(I_i^{(m)}\) cover one full period and
\(x_2=\sin\phi\) along the invariant curve, the intervals
\(A_i^{(m)}=\sin(I_i^{(m)})\) cover \([-1,1]\). It remains to check that passing from \(A_i^{(m)}\) to the trimmed intervals \(\boxB_i^{(m)}\) does not create gaps, and in fact leaves a uniformly positive overlap between consecutive trimmed intervals.

We first consider two consecutive non-wrapping phase windows whose images avoid $\pm1$ and whose union lies on
a monotone branch of \(\sin\):
\[
    I_i^{(m)}=[ih^{(m)},(i+2)h^{(m)}],
    \quad
    I_{i+1}^{(m)}=[(i+1)h^{(m)},(i+3)h^{(m)}].
\]
Suppose first that \(\sin\) is increasing on
\([ih^{(m)},(i+3)h^{(m)}]\). Then
\[
    A_i^{(m)}
    =
    [\sin(ih^{(m)}),\sin((i+2)h^{(m)})],
    \quad
    A_{i+1}^{(m)}
    =
    [\sin((i+1)h^{(m)}),\sin((i+3)h^{(m)})].
\]
Using
\[
    \boxB_j^{(m)}
    =
    \left[\alpha_j^{(m)}+\frac{\Delta_j^{(m)}}{4},
        \beta_j^{(m)}-\frac{\Delta_j^{(m)}}{4}
    \right]
    =
    \left[
    \frac{3\alpha_j^{(m)}+\beta_j^{(m)}}{4},
    \frac{\alpha_j^{(m)}+3\beta_j^{(m)}}{4}
    \right]
\]
for $j\in\{i,i+1\}$, the overlap length is
\[
\begin{aligned}
    \left|
        \boxB_i^{(m)}
        \cap
        \boxB_{i+1}^{(m)}
    \right|
    &=
    \frac{
        \sin(ih^{(m)})
        +
        3\sin((i+2)h^{(m)})
        -
        3\sin((i+1)h^{(m)})
        -
        \sin((i+3)h^{(m)})
    }{4}                                                   \\
    &=
    2\sin^3\left(\frac{h^{(m)}}{2}\right)
    \cos\left(\left(i+\frac32\right)h^{(m)}\right).
\end{aligned}
\]
On a monotone increasing branch, the cosine factor is nonnegative. Moreover, for
the monotone pairs under consideration, its smallest possible value occurs for the pair closest to a critical point of \(\sin\), and is at least \(\sin(3h^{(m)}/2)\). Therefore,
\[
    \left|
        \boxB_i^{(m)}
        \cap
        \boxB_{i+1}^{(m)}
    \right|
    \ge
    2\sin^3\left(\frac{h^{(m)}}{2}\right)
    \sin\left(\frac{3h^{(m)}}{2}\right)
    =
    \omega^{(m)}.
\]
The case where \(\sin\) is decreasing on \([ih^{(m)},(i+3)h^{(m)}]\) is identical, with the order of the endpoints reversed. It remains to check the pairs near the extrema of \(\sin\). We treat the maximum \(1\); the minimum \(-1\) is symmetric. Let $r \defeq 2^m$. The phase windows whose \(A\)-intervals contain \(1\) are those with indices \(r-2\), \(r-1\), and \(r\). By the endpoint rule, these intervals are trimmed on
the side away from \(1\) and extended on the side beyond \(1\). Thus, \(1\) is strictly inside each of the corresponding \(\boxB\)-intervals, so these endpoint intervals overlap with one another. Their pairwise overlap has length at least $1-\cos h^{(m)}\geq\omega^{(m)}$, since $\omega^{(m)}=(1-\cos h^{(m)})\sin(h^{(m)}/2)\sin(3h^{(m)}/2)$.

We only need to check the overlap with the adjacent ordinary intervals. Consider
the left adjacent pair \(\boxB_{r-3}^{(m)}\) and \(\boxB_{r-2}^{(m)}\). Since
\(rh^{(m)}=\pi/2\), we have
\[
    A_{r-3}^{(m)}
    =
    [\cos(3h^{(m)}),\cos h^{(m)}],
    \quad
    A_{r-2}^{(m)}
    =
    [\cos(2h^{(m)}),1].
\]
Thus, the overlap length is at least
\[
    \frac{
        \cos(3h^{(m)})
        +
        3\cos h^{(m)}
        -
        3\cos(2h^{(m)})
        -
        1
    }{4}.
\]
Using the identity
\[
    \cos(3h^{(m)})
    +
    3\cos h^{(m)}
    -
    3\cos(2h^{(m)})
    -
    1
    =
    2(1-\cos h^{(m)})^2(2\cos h^{(m)}+1),
\]
this becomes
\[
    \frac12(1-\cos h^{(m)})^2(2\cos h^{(m)}+1)
    =
    2\sin^3\left(\frac{h^{(m)}}{2}\right)
    \sin\left(\frac{3h^{(m)}}{2}\right)
    =
    \omega^{(m)}.
\]
The right adjacent pair around the maximum is handled by symmetry around \(\pi/2\). The minimum \(-1\) is handled in the same way by symmetry around \(3\pi/2\). Finally, the wrapping pair is covered by the monotone-branch argument after identifying
\[
    I_{N^{(m)}-1}^{(m)}
    =
    [2\pi-h^{(m)},2\pi]\cup[0,h^{(m)}]
\]
with the interval \([-h^{(m)},h^{(m)}]\). Its union with
\(I_0^{(m)}=[0,2h^{(m)}]\) lies on a monotone branch of \(\sin\) for \(m\ge2\). The other pair involving the wrapping window follows by reflection about phase zero. For $m=1$, both overlaps involving the wrapping window have length $(\sqrt{2}-1)/4=\omega^{(1)}$.

Thus, every consecutive pair of trimmed intervals overlaps by at least \(\omega^{(m)}>0\). In particular, the intervals form an overlapping chain whose union contains both endpoints \(-1\) and \(1\). Therefore, they cover the trajectory's $x_2$-range: $[-1,1]\subseteq \bigcup_{i=0}^{N^{(m)}-1}\boxB_i^{(m)}$.
\end{proof}

The preceding lemma shows that the trimmed intervals
cover the trajectory's \(x_2\)-range with positive overlaps. Each interval
\(\boxB_i^{(m)}=[\boxalpha_i,\boxbeta_i]\) is paired with the auxiliary coordinate \(x_{3+i}\), for \(i\in\{0,\ldots,N^{(m)}-1\}\). In the augmented potential~\eqref{eq:interm-pot}, this coordinate is coupled to the quadratic function
\[
    q_i^{(m)}(x_2)
    \defeq
    (x_2-\boxalpha_i)(\boxbeta_i-x_2).
\]
Thus, whenever \(x_2\) lies in the interior of the interval \(\boxB_i^{(m)}\), the coordinate \(x_{3+i}\) receives a positive gradient: increasing \(x_{3+i}\) from the boundary value \(0\) is a feasible improving deviation. Outside \(\boxB_i^{(m)}\), the quadratic is always nonpositive; in fact, for $x_2\in[-1,1]\setminus A_i^{(m)}$, it is strictly negative. This structure is what allows the cumulative negative contribution to dominate the cumulative positive contribution in the dwell-time argument, as shown in \cref{sec:ftrl}.

It remains to show that, throughout the execution of \(\ftrl\), there are always variables \(x_{3+i}\) with a constant \(\KKTgap\).  The point of \cref{lem:trimmed_boxes_cover} is that such an index \(i\) always exists. Indeed, the trimmed intervals cover the trajectory's $x_2$-range, and consecutive intervals overlap by at least \(\omega^{(m)}>0\). Thus, \(x_2\)
cannot pass from one interval to the next only through points where all adjacent quadratics vanish. Moreover, at the endpoints \(\pm1\), our convention ensures that the corresponding endpoint term remains strictly positive. We split the proof into two parts. Both lemmas below concern the intermediate objective $\widetilde{\Phi}$ from~\eqref{eq:interm-pot}.

\begin{lemma}
\label{lem:kkt_gap_inside_box}
Fix \(m\ge1\) and \(i\in\{0,\ldots,N^{(m)}-1\}\). Suppose that the auxiliary coordinate \(x_{3+i}\) is kept at the boundary value \(0\). If
\(
    q_i^{(m)}(x_2(t))>0,
\)
then the coordinate-wise first-order gap of \(x_{3+i}\) is strictly positive:
\[
    \max_{x'_{3+i}\in\cX_{3+i}}
    \left\langle
        x'_{3+i}-x_{3+i},
        \nabla_{x_{3+i}}{\widetilde{\Phi}}(\vx(t))
    \right\rangle
    >0.
\]
\end{lemma}

\begin{proof}
By construction,
\[
    \nabla_{x_{3+i}} {\widetilde{\Phi}}(\vx(t))
    =
    q_i^{(m)}(x_2(t)).
\]
Since \(x_{3+i}=0\), choosing any feasible \(x'_{3+i}>0\) gives
\[
    \left\langle
        x'_{3+i}-x_{3+i},
        \nabla_{x_{3+i}} {\widetilde{\Phi}}(\vx(t))
    \right\rangle
    =
    x'_{3+i}q_i^{(m)}(x_2(t))
    >
    0,
\]
and the claim follows.
\end{proof}

\begin{lemma}[Uniformly positive \(\KKTgap\)]
\label{lem:uniform_kkt_gap}
Fix \(m\ge1\), and suppose that the auxiliary coordinates are kept at \(x_{3+i}=0\) for all \( i \in \{0,\ldots,N^{(m)}-1\}\). There is a constant $\gamma_m>0$ such that, at every time $t$,
\[
    \KKTgap(\vx(t)) \geq\gamma_m>0.
\]
\end{lemma}

\begin{proof}
Since the trajectory always satisfies \(x_2(t)\in[-1,1]\), it is enough to find one auxiliary coordinate whose quadratic gate is strictly positive. By \cref{lem:trimmed_boxes_cover}, the intervals \(\boxB_i^{(m)}\) cover \([-1,1]\), and consecutive intervals overlap by at least \(\omega^{(m)}>0\). Therefore, for every \(x_2(t)\in[-1,1]\), there exists an index
\(i\in\{0,\ldots,N^{(m)}-1\}\) such that
\(
    q_i^{(m)}(x_2(t))>0.
\)
For points in the interior of \([-1,1]\), this follows from the nontrivial overlap of consecutive trimmed intervals. At the endpoint values \(x_2(t)=\pm1\), it follows from our convention for the endpoint intervals.

We conclude that the continuous function $z\mapsto\max_i q_i^{(m)}(z)$ is positive on the compact interval $[-1,1]$, so $\gamma_m \defeq \min_{z\in[-1,1]} \max_i q_i^{(m)}(z)>0$. Choosing $x'_{3+i}=1$ for a maximizing index $i$ gives $\KKTgap(\vx(t))\geq\gamma_m$.
\end{proof}

\section{Convergence for analytic objectives}
\label{appendix:analytic}

Our counterexample relies on an infinitely differentiable objective that is not analytic. In this section, we show that, in contrast, analyticity \emph{does} guarantee convergence to a single KKT point. In particular, our goal is to establish the following.

\ifdefined\analyticconvergence
\analyticconvergence*
\else
\begin{theorem}
    \label{thm:analytic-convergence}
    Let $\cX \subseteq \R^d$ be a nonempty compact polytope, and let $\Phi$ be real analytic on an open neighborhood of $\cX$. Continuous-time $\ftrl$ with Euclidean regularization has finite trajectory length,
    \[
        \int_0^\infty \|\dot{\vx}(t)\|_2\,dt < \infty,
    \]
    and thereby converges to a point $\vx^\star \in \cX$ with $\KKTgap(\vx^\star)=0$.
\end{theorem}
\fi

The proof proceeds by relating the trajectory speed to the gradient along a certain face of the set, using basic projection geometry. Then, we use analyticity to show that the trajectory length is finite, as claimed in~\Cref{thm:analytic-convergence}.

In what follows, we write
\[
    N_{\cX}(\vx) \defeq \{\vn : \langle \vn,\vx'-\vx\rangle \leq 0 \text{ for all }\vx'\in\cX\}
\]
for the normal cone and
\[
    T_{\cX}(\vx) \defeq \overline{\{a(\vx'-\vx):a\geq0,\ \vx'\in\cX\}}
\]
for the tangent cone. For $\vx=\Pi_{\cX}(\vy)$, a standard property of the Euclidean projection implies that the residual $\vn=\vy-\vx$ belongs to $N_{\cX}(\vx)$. The associated \emph{critical cone} is $\cK=T_{\cX}(\vx)\cap\vn^\perp$, which retains the feasible directions orthogonal to the residual. The lemma below can be obtained as a special case of~\citet[Theorem 3.1]{Shapiro16:Differentiability}, specialized to polyhedral sets; we provide the simpler proof for completeness.

\begin{lemma}
    \label{lem:analytic-projection}
        For a nonempty closed polyhedron $\cX$, let $\vx=\Pi_{\cX}(\vy)$ and $\cK=T_{\cX}(\vx)\cap(\vy-\vx)^\perp$. Then, for every $\vv\in\R^d$,
    \[
        \lim_{h \downarrow 0}\frac{\Pi_{\cX}(\vy+h\vv)-\vx}{h}=\Pi_{\cK}(\vv).
    \]
\end{lemma}

\begin{proof}
    Let $\vn=\vy-\vx$ and $\vw_h=(\Pi_{\cX}(\vy+h\vv)-\vx)/h$. Nonexpansiveness of the Euclidean projection implies $\|\vw_h\|_2\leq\|\vv\|_2$, so every sequence $h_k \downarrow 0$ admits a subsequence along which $\vw_{h_k}$ converges. Moreover, $\langle\vn,\vw_h\rangle\leq0$ since $\vn \in N_{\cX}(\vx)$. Given that $\vx+h\vw_h = \Pi_{\cX}(\vy+h\vv)$ and $\vx\in\cX$, we have
    \[
        \|(\vy+h\vv)-(\vx+h\vw_h)\|_2^2
        \leq\|(\vy+h\vv)-\vx\|_2^2.
    \]
    Thus,
    \[
        0\leq-2\langle\vn,\vw_h\rangle
        \leq h\bigl(\|\vv\|_2^2-\|\vv-\vw_h\|_2^2\bigr) \leq h \|\vv\|_2^2,
    \]
    so every limit point of $\vw_h$ as $h\downarrow 0$ belongs to $\cK$. Now, for any $\vq\in \cK$, polyhedrality ensures that $\vx+h\vq\in\cX$ for all sufficiently small $h>0$: active linear constraints remain satisfied because $\vq$ belongs to the tangent cone, and inactive constraints have positive slack. Since $\vx+h\vw_h = \Pi_{\cX}(\vy+h\vv)$,
    \[
    \|(\vy+h\vv)-(\vx+h\vw_h)\|_2^2
    \leq
    \|(\vy+h\vv)-(\vx+h\vq)\|_2^2.
    \]
    Expanding,  
    \[
        -\frac{2}{h}\langle\vn,\vw_h\rangle+ \|\vv-\vw_h\|_2^2
        \leq\|\vv-\vq\|_2^2,
    \]
    so $\|\vv-\vw_h\|_2^2 \leq\|\vv-\vq\|_2^2$ since $\langle \vn, \vw_h \rangle \leq 0$. Thus, every limit point minimizes distance to $\vv$ over $\cK$, and the claim follows.
\end{proof}

Throughout this section, the nonempty faces of $\cX$ include $\cX$ itself, as well as all its proper nonempty faces. For each such face $F$, let $\cL_F\defeq\operatorname{span}(F-F)$, where $F-F=\{\vu-\vv:\vu,\vv\in F\}$. $\cL_F$ is the linear subspace of directions parallel to $F$. We define $\nabla_F\Phi(\vx)\defeq\Pi_{\cL_F}(\nabla\Phi(\vx))$, the gradient of $\Phi$ restricted to the affine hull of $F$. The relative interiors of the nonempty faces partition $\cX$, so each $\vx\in\cX$ lies in the relative interior of exactly one face, denoted by $F(\vx)$. Equivalently, $F(\vx)$ is the smallest face of $\cX$ containing $\vx$, in the sense that $F(\vx) \subseteq F$ for every face of $\cX$ containing $\vx$. In particular, $F(\vx)=\cX$ whenever $\vx$ lies in the relative interior of $\cX$.

\begin{lemma}
    \label{lem:analytic-speed}
    Under the assumptions of \Cref{thm:analytic-convergence}, for almost every $t \geq 0$,
    \[
        \|\dot{\vx}(t)\|_2\geq\|\nabla_{F(\vx(t))}\Phi(\vx(t))\|_2.
    \]
\end{lemma}

\begin{proof}
    The vector field $\vy \mapsto\nabla \Phi(\Pi_{\cX}(\vy))$ is Lipschitz and bounded, so there is a unique solution $\vy(t)$, and the projected trajectory $\vx(t)=\Pi_{\cX}(\vy(t))$ is Lipschitz. By \Cref{lem:analytic-projection} and nonexpansiveness, for almost every $t$,
    \[
        \dot{\vx}(t)=\Pi_{\cK(t)}(\nabla\Phi(\vx(t))),
    \]
    where $\cK(t)=T_{\cX}(\vx(t))\cap(\vy(t)-\vx(t))^\perp$. The first-order optimality condition yields
    \[
        \left\langle\nabla\Phi(\vx(t))-\dot{\vx}(t),\,
        \vv-\dot{\vx}(t)\right\rangle\leq0
        \quad\text{for every }\vv\in\cK(t).
    \]
    Since $\cK(t)$ is a cone, both $\vec{0}$ and $2\dot{\vx}(t)$ belong to it. Substituting these two choices for $\vv$ implies
    \[
        \left\langle\nabla\Phi(\vx(t))-\dot{\vx}(t),\,
        \dot{\vx}(t)\right\rangle=0.
    \]
    By the chain rule, this yields $\frac{d}{dt}\Phi(\vx(t))=\|\dot{\vx}(t)\|_2^2$ almost everywhere.

    We next argue that $\cL_{F(\vx(t))}\subseteq\cK(t)$. Since $\vx(t)$ lies in the relative interior of $F(\vx(t))$, for every $\vu \in \cL_{F(\vx(t))}$, the points $\vx(t) \pm \epsilon\vu$ belong to this face for sufficiently small $\epsilon>0$. Thus, $\pm\vu\in T_{\cX}(\vx(t))$. Moreover, since $\vy(t)-\vx(t) \in N_\cX(\vx(t))$, it follows that $\langle\vy(t)-\vx(t),\vu\rangle = 0$, so both $\vu$ and $-\vu$ belong to $\cK(t)$. Substituting $\vv=\vu$ and $\vv=-\vu$ into the first-order optimality condition,
    \[
        \left\langle\nabla\Phi(\vx(t))-\dot{\vx}(t),\,\vu\right\rangle\leq0,
        \quad
        -\left\langle\nabla\Phi(\vx(t))-\dot{\vx}(t),\,\vu\right\rangle\leq0.
    \]
    Therefore, $\nabla\Phi(\vx(t))-\dot{\vx}(t)$ is orthogonal to every vector in $\cL_{F(\vx(t))}$. Thus,
    \[
        \nabla_{F(\vx(t))}\Phi(\vx(t))
        =\Pi_{\cL_{F(\vx(t))}}\bigl(\nabla\Phi(\vx(t))\bigr)
        = \Pi_{\cL_{F(\vx(t))}}\bigl(\dot{\vx}(t)\bigr).
    \]
    We conclude that $\| \dot{\vx}(t) \|_2 \geq \| \Pi_{\cL_{F(\vx(t))}}\bigl(\dot{\vx}(t)\bigr) \|_2 = \| \nabla_{F(\vx(t))}\Phi(\vx(t)) \|_2 $.
\end{proof}

Analyticity enters through the \L{}ojasiewicz \emph{gradient inequality}~\citep{Lojasiewicz84:Trajectoires,Absil05:Convergence}. Its local form states that, near a critical point $\vp$ of a real analytic function $f$, there are parameters $c>0$ and $\theta\in[1/2,1)$ such that $\|\nabla f(\vx)\|_2\geq c|f(\vx)-f(\vp)|^\theta$. Using this crucial property, we obtain the following.

\begin{lemma}
    \label{lem:analytic-uniform}
    For every $\ell\in\R$, there exist $c>0$, $\eta\in(0,1]$, and $\theta\in[1/2,1)$ such that, for every nonempty face $F$ and every $\vx\in F$ with $0<|\Phi(\vx)-\ell|<\eta$,
    \[
        \|\nabla_F\Phi(\vx)\|_2\geq c|\Phi(\vx)-\ell|^\theta.
    \]
\end{lemma}

\begin{proof}
    For a nonempty face $F$, we consider the compact level set $S_F\defeq\{\vp\in F:\Phi(\vp)=\ell\}$. If $S_F$ is empty, we can set $c_F=1$, $\theta_F=1/2$, and $\eta_F=\min\{1,\tfrac12\min_{\vx\in F}|\Phi(\vx)-\ell|\}>0$; the assertion is then vacuous. If $F$ is a vertex and $S_F$ is nonempty, we can take $(c_F,\eta_F,\theta_F)=(1,1,1/2)$, again giving a vacuous assertion. Otherwise, we first establish gradient bounds near each point of $S_F$, and then show that these bounds cover all points of $F$ whose potential value is sufficiently close to $\ell$.

    At each $\vp\in S_F$ that is a critical point of the
    restriction of $\Phi$ to the affine hull of $F$, the
    \L{}ojasiewicz gradient inequality gives $\|\nabla_F\Phi(\vx)\|_2
        \geq c_\vp|\Phi(\vx)-\Phi(\vp)|^{\theta_\vp}
        =c_\vp|\Phi(\vx)-\ell|^{\theta_\vp}$ on a relatively open neighborhood of $\vp$, with
    $c_\vp>0$ and $\theta_\vp\in[1/2,1)$. At each noncritical point $\vp\in S_F$, continuity gives
    a relatively open neighborhood on which
    $\|\nabla_F\Phi\|_2 \geq c_\vp>0$. Shrinking this neighborhood so that $|\Phi-\ell|<1$ yields the same inequality with $\theta_\vp=1/2$. By compactness, finitely many of these neighborhoods cover $S_F$. Let $c_F$ be the minimum among $\{ c_{\vp} \}$ and $\theta_F$ the maximum among $\{ \theta_\vp \}$. We have $\|\nabla_F\Phi(\vx)\|_2 \geq c_F|\Phi(\vx)-\ell|^{\theta_F}$ throughout their union whenever $0<|\Phi(\vx)-\ell|<1$. By compactness, we can choose $\eta_F\in(0,1]$ so that every $\vx\in F$ with $|\Phi(\vx)-\ell|<\eta_F$ belongs to one of the chosen neighborhoods. This proves the desired bound on $F$. Finally, since $\cX$ has finitely many nonempty faces, we may take $c\defeq \min_F c_F>0$, $\eta \defeq\min_F\eta_F>0$, and $\theta\defeq\max_F\theta_F\in[1/2,1)$, which give the asserted
    inequality simultaneously on every nonempty face.
\end{proof}

\begin{proof}[Proof of \Cref{thm:analytic-convergence}]
    By \Cref{lem:analytic-speed} and compactness, $\Phi(\vx(t))\uparrow\ell$ for some $\ell\in\R$. We set $E(t)\defeq\ell-\Phi(\vx(t))$, so $\dot E(t)=-\|\dot{\vx}(t)\|_2^2$ almost everywhere. For all sufficiently large $t$, \Cref{lem:analytic-uniform,lem:analytic-speed} give
    \[
        \|\dot{\vx}(t)\|_2\geq cE(t)^\theta
    \]
    almost everywhere whenever $E(t)>0$. Thus,
    \[
        -\frac{d}{dt}E(t)^{1-\theta}
        =(1-\theta)E(t)^{-\theta}\|\dot{\vx}(t)\|_2^2
        \geq c(1-\theta)\|\dot{\vx}(t)\|_2.
    \]
    Integrating bounds the remaining path length of $\vx(t)$ by $E(t)^{1-\theta}/(c(1-\theta))$. If $E(t)$ reaches zero, it follows that $\vx(t)$ will be constant thereafter, and the same bound follows by integration up to that time. Since $\vx(t)$ is Lipschitz, we conclude that its total length is finite. It therefore converges to some $\vx^\star\in\cX$, with
    \[
        \|\vx(t)-\vx^\star\|_2
        \leq\int_t^\infty\|\dot{\vx}(\tau)\|_2\,d \tau
        \leq\frac{E(t)^{1-\theta}}{c(1-\theta)}
    \]
    for all sufficiently large $t$. Finally, to establish stationarity, we use continuity of $\nabla\Phi$ to obtain
    \[
        \frac{\vy(t)}{t}
        =\frac{\vy_0}{t}+\frac{1}{t}\int_0^t \nabla\Phi(\vx(\tau))\, d\tau
        \longrightarrow\nabla\Phi(\vx^\star).
    \]
    Taking the limit yields $\langle\nabla\Phi(\vx^\star),\vx'-\vx^\star\rangle\leq0$ for every $\vx'\in\cX$. Thus, $\KKTgap(\vx^\star)=0$.
\end{proof}
    
In particular, \Cref{thm:analytic-convergence} applies to compact boxes, so our nonconvergent construction cannot be made analytic while retaining Euclidean regularization. It also establishes convergence for finite potential games.

\ifdefined\analyticgames
\analyticgames*
\else
\begin{corollary}
    \label{cor:analytic-games}
    In every finite potential game, continuous-time $\ftrl$ with Euclidean regularization and arbitrary initialization has finite trajectory length and converges to a Nash equilibrium.
\end{corollary}
\fi


\end{document}